\documentclass[aps,pra,10pt,groupedaddress,notitlepage,showpacs,floatfix,longbibliography,superscriptaddress,twocolumn]{revtex4-1}
\pdfoutput=1
\usepackage{graphicx,graphics,epsfig,subfigure,times,bm,bbm,amssymb,amsmath,amsfonts,amsthm,mathrsfs,MnSymbol}
\usepackage{gensymb}
\usepackage[matrix,frame,arrow]{xypic}
\usepackage[pdfstartview=FitH]{hyperref}
\usepackage{subfigure}
\hypersetup{
    colorlinks=true,       % false: boxed links; true: colored links
    linkcolor=red,          % color of internal links
   citecolor=magenta,        % color of links to bibliography
    filecolor=magenta,      % color of file links
    urlcolor=cyan,           % color of external links
    runcolor=cyan
}
\usepackage[pdftex]{color}
\usepackage{braket}%Dirac Notation in QM
\usepackage{enumerate}
\usepackage[normalem]{ulem}
\usepackage[usenames,dvipsnames]{xcolor}
\usepackage{multirow}
\usepackage{mathtools}

\definecolor{orange}{rgb}{1,0.5,0}

\newtheorem{theorem}{Theorem}[section]
\newtheorem{corollary}{Corollary}[theorem]

\usepackage{tikz}
\newcommand{\bes} {\begin{subequations}}
\newcommand{\ees} {\end{subequations}}
\newcommand{\bea} {\begin{eqnarray}}
\newcommand{\eea} {\end{eqnarray}}

\definecolor{gold}{rgb}{0.85,.66,0}

\newcommand{\beq}{\begin{equation}}

\newcommand{\eeq}{\end{equation}}

\newcommand{\ignore}[1]{}

\def\g{\gamma}

\def\s{\sigma}

\def\>{\rangle}
\def\<{\langle}

\def\s0{I}

\newcommand{\ig}[1]{}

\definecolor{jmlBlue}{RGB}{0,163,224}

\begin{document}
\title{Noiseless Loss Suppression for Entanglement Distribution}
\author{Cory M. Nunn}
\affiliation{National Institute of Standards and Technology, Gaithersburg, Maryland 20899, USA}
\affiliation{University of Maryland Baltimore County, Baltimore, MD 21250, USA}
\author{Daniel E. Jones}
\affiliation{DEVCOM Army Research Laboratory, Adelphi, MD 20783, USA}
\affiliation{National Institute of Standards and Technology, Gaithersburg, Maryland 20899, USA}
\author{Todd B. Pittman}
\affiliation{University of Maryland Baltimore County, Baltimore, MD 21250, USA}
\author{Brian T. Kirby}
\affiliation{DEVCOM Army Research Laboratory, Adelphi, MD 20783, USA}
\affiliation{Tulane University, New Orleans, LA 70118, USA}
\date{\today}
\begin{abstract}
Recent work by Mi\v{c}uda et al.~\cite{mivcuda2012noiseless} suggests that pairing noiseless amplification with noiseless attenuation can conditionally suppress loss terms in the direct transmission of quantum states. 
Here we extend this work to entangled states: first, we explore bipartite states, specifically the two-mode squeezed vacuum (TMSV) and NOON states; and second, we examine $M$-partite states, concentrating on W and Greenberger–Horne–Zeilinger (GHZ) states.
In analogy with the original proposal, our results demonstrate that in each case under consideration, a correct combination of attenuation and amplification techniques before and after transmission through a pure loss channel can restore the initial quantum state.
However, we find that for both W and NOON states, the noiseless attenuation is redundant and not required to achieve loss term suppression. This work clarifies the role of noiseless attenuation when paired with noiseless amplification for entanglement distribution and provides an operational example of how GHZ and W state entanglement differs. 
\end{abstract}
\maketitle

\section{Introduction}

Quantum networking offers the potential for various advanced applications such as distributed quantum computing \cite{cirac1999distributed}, long-baseline interferometry \cite{gottesman2012longer}, and precise time synchronization \cite{komar2016quantum}. Photon loss in a quantum network's channels, links, and components limits the ability to realize these applications in practice.
Traditional methods to counter these issues, such as quantum error correction and quantum repeaters, can have demanding technical requirements including the need for quantum memories, multiple copies of a quantum state, or quantum information processing at the nodes. 
Unlike classical communication, the accumulated photon losses in a quantum network cannot be compensated for with phase-insensitive amplification due to the inherent noise added and the no-cloning principle~\cite{pandey2013quantum}.
However, these limitations can be partially mitigated using heralded noiseless amplification of light~\cite{ralph2009nondeterministic}, which conditionally boosts the amplitude of a coherent state without adding noise, presenting a promising avenue for overcoming these challenges.

%%%%
Noiseless linear amplification (NLA) has been used in a variety of schemes to reduce the effects of loss on quantum communication for both continuous-variable \cite{xiang2010heralded} and discrete \cite{osorio2012heralded} encodings. For instance, it has been used to probabilistically amplify coherent states \cite{xiang2010heralded,ferreyrol2010implementation,usuga2010noise,zavatta2011high,donaldson2015experimental,haw2016surpassing} and two-mode squeezed vacuum (TMSV) \cite{ulanov2015undoing,dias2020quantum,karsa2022noiseless}, as well as polarization \cite{kocsis2013heralded}, path \cite{monteiro2017heralded}, and time-bin states \cite{bruno2016heralded}. In addition, various schemes utilizing NLA have been used to perform entanglement distillation~\cite{mauron2022comparison,ulanov2015undoing,monteiro2017heralded,dias2020quantum,liu2022distillation,slussarenko2022quantum}.
While noiseless amplification alone can improve the performance of quantum state distribution systems, some schemes also utilize quantum scissors \cite{seshadreesan2019continuous,dias2020quantum,slussarenko2022quantum} or photon subtraction techniques \cite{usuga2010noise,zavatta2011high,donaldson2015experimental}.
Unfortunately, NLA alone is ultimately limited in its ability to overcome loss in the direct transmission of superpositions of Fock states due to the inevitable mixing incurred by these channels. 
More precisely, the NLA process is not the inverse map of the loss process, and hence it cannot eliminate the noise added by it~\cite{mivcuda2012noiseless}. 

An alternative approach devised by Mi\v{c}uda et al.\@ supplements the use of NLA with a companion process called noiseless \textit{attenuation}~\cite{mivcuda2012noiseless}. In the proposed noiseless loss suppression (NLS) protocol, quantum states are noiselessly attenuated prior to their distribution across a lossy channel, and then noiseless amplification restores their amplitudes at the end of the channel. Crucially, this type of heralded noiseless attenuation can be faithfully inverted by the noiseless amplifier, as these operations are inverse maps of one another. When the pre-attenuation is sufficiently strong such that the effects of photon loss are negligible in comparison, the full NLS protocol can transmit quantum states with arbitrarily high fidelity at the cost of reduced probability of success. Mi\v{c}uda et al. showed that, in principle, any quantum state from an arbitrarily large Hilbert space could be transmitted across a single lossy channel in this way.
Here, we investigate the potential use of NLS for entanglement distribution of multi-partite entanglement through independent lossy channels.

We begin in Section \ref{sec:Micuda} by reviewing the original NLS protocol of Mi\v{c}uda et al.\@ for a single optical mode with states encoded in a superposition of vacuum and single-photon states. 
Then, to provide insight into our main results, in Section \ref{sec:experiment} we consider the application of the NLS protocol to bipartite entangled states by introducing a possible experimental setup involving a bipartite entangled state consisting of only a single photon with physical noiseless amplifiers based on quantum scissors.
We then abstract away the experimental configuration and generalize these bipartite states to $M$ modes sharing a single photon (i.e., W states) in Section \ref{sec:W} and to $n$-photon Greenberger–Horne–Zeilinger (GHZ) states in Section \ref{sec:GHZ}. Next, in Section \ref{sec:TMSV} we consider the application of NLS to TMSV states where higher-order photon terms complicate the analysis. Then in  Section \ref{sec:noon} we again consider multiple photons per mode by studying the NLS protocol on NOON states. In Section \ref{sec:necessity} we provide general arguments describing which states require noiseless attenuation in addition to noiseless amplification in order to recover their initial form, followed by a discussion and conclusion in Section \ref{sec:discussion}.

\section{Noiseless Loss Suppression}
\label{sec:Micuda}

The original NLS protocol devised by Mi\v{c}uda et al.~\cite{mivcuda2012noiseless} is shown in Fig.~\ref{fig:NLSdiagram}(a). Here, an input state $\hat{\rho}_{in}$ is sent across a single lossy channel $\mathcal{L}_{\tau}$ with amplitude transmittance $\tau$. Unaided, photon loss over this channel will generate noise and reduce the coherence of the transmitted output $\hat{\rho}_{out}$. The NLS scheme mitigates these effects with the introduction of two probabilistic devices: the noiseless attenuator and the noiseless amplifier. The noiseless attenuation is described by the quantum filter $\nu^{\hat{n}}$ where $\hat{n}=\hat{a}^\dagger\hat{a}$ is the number operator and $\nu<1$, which reduces Fock state amplitudes by $\ket{n}\rightarrow \nu^n \ket{n}$. This precedes the lossy channel, essentially suppressing the higher-photon-number terms which are most susceptible to photon loss. A noiseless amplifier, which similarly transforms Fock state amplitudes as $\ket{n}\rightarrow g^n\ket{n}$ with $g>1$, is then placed at the end of the channel with gain $g=1/(\nu\tau)$ chosen to restore the original amplitudes. 
In the limit $\nu\rightarrow 0$, noise terms generated by photon loss are eliminated, and the input state is faithfully (but conditionally) transmitted, i.e. $\hat{\rho}_{out}=\hat{\rho}_{in}$.

\begin{figure}[t]
    \centering
    \includegraphics[width=0.9\columnwidth,trim={250 30 240 100},clip]{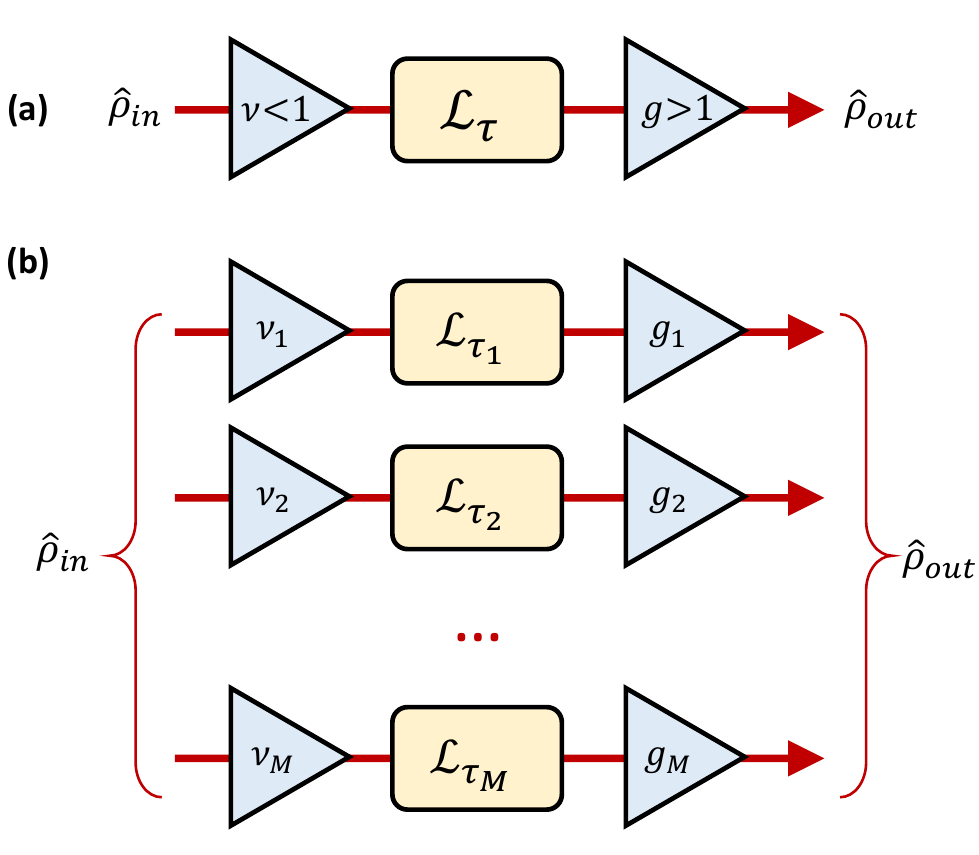}
    \caption{Basic illustration of the noiseless loss suppression (NLS) protocol~\cite{mivcuda2012noiseless}. Panel (a) shows the case of a single-mode input state $\hat{\rho}_{in}$ transmitted over one quantum channel, while (b) shows an extension of NLS to multimode states sent over $M$ channels. Each channel is preceded by a noiseless attenuator with an attenuation factor $\nu_k<1$, and followed with noiseless amplification by gain $g_k>1$. With appropriate choices of attenuation and gain, any added noise generated by the lossy channels $\mathcal{L}(\tau_k)$ can be mitigated.}
    \label{fig:NLSdiagram}
\end{figure}

Following Ref.~\cite{mivcuda2012noiseless}, we will illustrate the mechanics of NLS in more detail with a simple example: the single-rail ($M=1$) qubit $\ket{\psi} =c_{0}\ket{0} + c_{1}\ket{1}$. Sending this state through the lossy channel $\mathcal{L}_{\tau}$, we obtain a mixed output with added noise:
\begin{equation}
    \mathcal{L}_{\tau}\left(\ket{\psi}\bra{\psi}\right) = \ket{\tilde{\psi}} \bra{\tilde{\psi}} + (1-\tau^{2})\vert c_{1}\vert^{2} \ket{0}\bra{0}
\end{equation}
where $\ket{\tilde{\psi}} = c_{0}\ket{0} + \tau c_{1}\ket{1}$. To restore the single-photon amplitude, we can noiselessly amplify this state with gain $g=1/\tau$. This transforms the Fock state amplitudes as $\ket{0}\rightarrow \ket{0}$ and $\ket{1}\rightarrow g\ket{1}$, which yields the following output:
\begin{equation}
    \hat{\rho}_{amp}\propto \ket{\psi}\bra{\psi} + (1-\tau^{2})\vert c_{1}\vert^{2} \ket{0}\bra{0}.
\end{equation}
The original signal has been restored, but the extra vacuum noise term remains. Further amplification would reduce the size of the noise term relative to the signal but would over-amplify the single-photon amplitude and alter the $\ket{\psi}\bra{\psi}$ term in the mixture. 

However, this can be corrected if the lossy channel $\mathcal{L}_\tau$ is preceded by a noiseless attenuator. Prior to $\mathcal{L}_\tau$, we now have a pure attenuated state $\ket{\psi'}=c_0\ket{0}+\nu c_1\ket{1}$. Then after amplifying by $g=1/(\nu\tau)$ at the end of the channel, we obtain:
\begin{equation}
    \hat{\rho}_{out}\propto \ket{\psi}\bra{\psi} + (1-\tau^{2})\nu^2\vert c_{1}\vert^{2} \ket{0}\bra{0}.
\end{equation}
Now the vacuum noise term is suppressed by the attenuation factor $\nu^2$. 
In the limit $\nu\rightarrow 0$, which requires more extreme attenuation and amplification, arbitrarily high fidelity can be achieved.
Note, however, that both noiseless attenuation and heralded amplification have a success rate that scales with the attenuation and gain introduced into the channel. Hence, the NLS protocol introduces a trade-off between channel fidelity and success rate, with an exact single-mode NLS success probability expression provided in \cite{mivcuda2012noiseless}.

The benefits of NLS still hold true for any superposition of Fock states in a single mode, so long as noiseless amplification can be faithfully implemented at the end of the channel. Because the ideal noiseless amplification operator $g^{\hat{n}}$ is unbounded, this requires that we restrict our input states to a finite subspace spanned by Fock states $\ket{n}$ for $n\leq N$ where $N$ is the maximum photon number in the channel~\cite{pandey2013quantum,gagatsos2014heralded}. 
Then for the full sequence of noiseless attenuation, loss, and then noiseless amplification shown in Figure~\ref{fig:NLSdiagram}(a), this complete channel $\mathcal{M}$ can be described as follows:
\begin{equation}
    \mathcal{M}(\hat{\rho}_{in}) = \sum_{j=0}^{N} G_N(g) \hat{A}_j \nu^{\hat{n}} \hat{\rho}_{in} \nu^{\hat{n}} \hat{A}^\dagger_j G_N(g)
\end{equation}
where the Kraus operators $\hat{A}_j$ are given by:
\begin{equation}
    \hat{A}_j = \sum_{m=0}^{N-j}\sqrt{\binom{m+j}{j}}(1-\tau^2)^{j/2}\tau^m\ket{m}\bra{m+j}
\end{equation}
and each correspond to the loss of $j$ photons from the input state. Noiseless amplification is implemented by the filter ${G_N(g)=g^{-N}\sum_{n=0}^{N} g^n \ket{n}\bra{n}}$, which replicates the action of the operator $g^{\hat{n}}$ within the smaller Hilbert space of states with a maximum photon number $N$. The success probability for this filter $G_N$ is bounded below by $g^{-2N}$. 

Following Ref.~\cite{mivcuda2012noiseless}, it is straightforward to show that with a choice of gain $g=1/(\nu\tau)$, the resulting output state takes the following form:
\begin{equation}\label{eqn:NLSch}
    \hat{\rho}_{out}\propto \hat{\rho}_{in} + \sum_{j=1}^N{\nu^{2j}\hat{B}_j \hat{\rho}_{in} \hat{B}^\dagger_j}
\end{equation}
with newly defined composite operators $\hat{B}_j=\hat{A}_0^{-1}\hat{A}_j$. The inverse Kraus operator $\hat{A}_0^{-1}=\sum_{m=0}^{N}\tau^{-m}\ket{m}\bra{m}$ is equivalent to noiseless amplification by $g=1/\tau$, much like $\hat{A}_0$ is equivalent to noiseless attenuation by $\nu=\tau$. As these are exact inverses of one another, the lossless $j=0$ term leaves the input state unchanged. Meanwhile, the remaining noise terms corresponding to $j>0$ photon losses are suppressed by a factor $\nu^{2j}$. In the single-photon example, the loss of $j=1$ photons produced a single vacuum noise term, which acquired an additional factor of $\nu^2$. More generally, this channel tends to the identity $\mathcal{M}\rightarrow\mathcal{I}$ in the limit $\nu\rightarrow0$, and input states can be transmitted with arbitrarily high fidelity.

Here we propose a natural extension of NLS to entangled multimode states, shown in Figure~\ref{fig:NLSdiagram}(b). The input $\hat{\rho}_{in}$ propagates along $M$ lossy paths, each with a distinct transmittance $\tau_k$ for $k=1,...,M$. As pictured, the original protocol is repeated for every mode: each path $k$ is preceded by a noiseless attenuator with $\nu_k<1$ and followed by an amplifier with $g_k>1$. The operators for mode $k$ commute with those of all other modes, so we can apply the NLS protocol $M$ times successively, once per each mode. All together, this transforms the state as:
\begin{equation}
\begin{aligned}
    &\mathcal{M}_{tot}(\hat{\rho}_{in}) =\\ &\sum_{j_1=0}^N \cdots \sum_{j_M=0}^N \left(\prod_{k=0}^M g_k^{-2N} \nu_k^{2j_k} \right) \hat{B}^{(M)}_{j_M} \cdots \hat{B}^{(1)}_{j_1} \hat{\rho}_{in} \hat{B}^{(1)\dagger}_{j_1} \cdots \hat{B}^{(M)\dagger}_{j_M}\\
    \end{aligned}
    \label{eq:general_outcome}
\end{equation}
where the operators $\hat{B}^{(k)}_{j_k}$ are similarly defined for each mode $k$:
\begin{equation}\label{eqn:bk_ops}
    \hat{B}^{(k)}_{j_k} = \sum_{m=j}^{N}\sqrt{\binom{m}{j}}(1-\tau_k^2)^{j_k/2}(\nu_k \tau_k g_k )^{m-j_k}\ket{m-j_k}\bra{m}
\end{equation}
Note, that unlike in the definition of $\hat{B}_{j}$ in the single mode case (which follows the treatment in \cite{mivcuda2012noiseless}), this definition of a multimode $\hat{B}_{jk}^{(k)}$ does not enforce any particular relationship between $g_k$ and $\nu_k$ (i.e., they are free variables).

The resulting output state from the composite channel can then be written in the following form:
\begin{equation}\label{eqn:NLStot}
    \hat{\rho}_{out} \propto \tilde{\rho}_{0\cdots 0} + \sum_{\bm{j}\neq \{0\}} \left(\prod_{k=1}^M{\nu_k^{2j_k}}\right)\tilde{\rho}_{j_1\cdots j_M}
\end{equation}
where the separate sums over $j_k$ in Eq.~\ref{eq:general_outcome} have been replaced with a combined sum over all possible sequences ${\bm{j}=\{j_1,...,j_M\}}$, and we have defined the following (unnormalized) density operators:
\begin{equation}\label{eqn:rhonoise}
    \tilde{\rho}_{j_1\cdots j_m} \equiv \hat{B}^{(M)}_{j_M} \cdots \hat{B}^{(1)}_{j_1} \hat{\rho}_{in} \hat{B}^{(1)\dagger}_{j_1} \cdots \hat{B}^{(M)\dagger}_{j_M}
\end{equation}
The result of Eq.~\ref{eqn:NLStot} has been organized into two separate terms: (i) a \textit{signal} term $\tilde{\rho}_{0\cdots 0}$ corresponding to zero photons lost in all modes, denoted ${\bm{j}=\{0\}}$; and (ii) a collection of \textit{noise} terms $\tilde{\rho}_{j_1\cdots j_M}$ corresponding to all other unique sequences of photon loss with some $j_k>0$. As in the original NLS protocol, the goal is to suppress the noise relative to the signal while also ensuring maximum fidelity to the input state.

With the greater number of degrees of freedom in this multimode system, more investigation is required to find the optimal values of attenuation and amplification for NLS.
One possible strategy is to mimic the original NLS protocol and choose ${g_k=1/(\nu_k\tau_k)}$ for each mode.
With this choice of parameters, each $\hat{B}^{(k)}_{0}$ is equal to the identity when $j_k=0$ photons are lost. Thus from Eq.~\ref{eqn:rhonoise}, when $j_k=0$ for all $k$ modes, the signal is left unchanged from the input state: $\tilde{\rho}_{0\dots 0} = \hat{\rho}_{in}$.

Though one might suspect otherwise, this naive approach works perfectly well despite any path entanglement that may exist in the system. Following Eq.~\ref{eqn:NLSch}, each instance of $\mathcal{M}$ returns the original state $\hat{\rho}_{in}$ with added noise. The total channel $\mathcal{M}_{tot}$ in Eq.~\ref{eqn:NLStot} thus has a similar form. Now every noise term corresponding to $j_k$-photon losses in mode $k$ will be suppressed by a factor of $\nu_k^{2j_k}$, with the original state corresponding to the unique case of $j_k=0$ photons lost in all $M$ modes. Therefore, even if the $M$ modes are entangled, we can still compensate for photon loss locally with $M$ pairs of noiseless attenuators and amplifiers.

However, while the naive approach to NLS will work for any state, it is not the best solution in every case. In the remainder of this paper, we will examine the performance of NLS for several common forms of entangled input states and consider the appropriate values of attenuation and gain (e.g., look for solutions where $g_k\ne 1/(\nu_k\tau_k)$). Interestingly, there are some cases in which NLS can be accomplished while omitting noiseless attenuation, allowing for compensation with only noiseless amplifiers at the end of each channel.

\section{Bipartite single-photon entanglement}
\label{sec:experiment}

Here we consider the most straightforward extension of the original NLS protocol to entangled states: its application to the superposition of a single photon in two spatial modes. Practically, single-photon entangled states of this type result from transmitting a single photon through a beamsplitter. In analogy with the single-mode states considered in the original paper, the vacuum modes of single-photon entangled states encode quantum information. 
Single-photon entanglement of this type has applications throughout quantum information science and has motivated fundamental questions about whether its nonlocal properties are detectable in experiments \cite{tan1991nonlocality,dunningham2007nonlocality,das2021can}. Note that this state has been considered in a related context in both \cite{monteiro2017heralded,slussarenko2022quantum}.

To develop physical intuition for our more generalized results to follow, we will now examine a specific proposed optical experiment and examine the properties of the output state.
The proposed experiment is pictured in Figure~\ref{fig:bothchannels} and consists of an $M=2$ mode experiment.
Here a single photon is prepared in an equal superposition of the two optical modes in the state $\vert \psi\rangle=(1/\sqrt{2})(\vert 01\rangle+\vert 10\rangle)$ (i.e, a single photon is passed through a 50/50 beamsplitter, up to a phase).
The original NLS protocol \cite{mivcuda2012noiseless} is then applied independently in each mode in parallel.
At the beginning of each channel, noiseless attenuation is applied by using beamsplitters with amplitude transmittance $\nu_{1}$ and $\nu_{2}$ and conditioning on the detection of zero photons in the two detectors ~\cite{nunn2021heralding}, resulting in: 
\begin{equation}
    \ket{\psi'} \propto \nu_1\ket{10} + \nu_2\ket{01}.
\end{equation}
Channel loss is then applied via another set of beamsplitters, each with transmittance given by $\tau_{1,2}$. In contrast to noiseless attenuation, the output modes are not detected and will be traced over, generating a mixed state with the noise terms that NLS aims to suppress. Before tracing over the loss modes, the state is given by:
\begin{widetext}
\begin{equation}
\begin{aligned}
    \ket{\psi''} &\propto \nu_1\left(i\sqrt{1-\tau_{1}^{2}}\ket{00}\ket{1}_{R}\ket{0}_{R'} + \tau_{1}\ket{10}\ket{0}_{R}\ket{0}_{R'} \right) + \nu_2\left(i\sqrt{1-\tau_{2}^{2}}\ket{00}\ket{0}_{R}\ket{1}_{R'} + \tau_{2}\ket{01}\ket{0}_{R}\ket{0}_{R'} \right) \\
    &= \left(\nu_{1}\tau_{1}\ket{10}+\nu_{2}\tau_{2}\ket{01}\right)\ket{0}_{R}\ket{0}_{R'}+i\ket{00}\left(\nu_{1}\sqrt{1-\tau_{1}^{2}}\ket{1}_{R}\ket{0}_{R'}+\nu_{2}\sqrt{1-\tau_{2}^{2}}\ket{0}_{R}\ket{1}_{R'}\right)
\end{aligned}
\end{equation}
where $R$ and $R'$ indicate the reflected modes of the channel loss beampslitters in modes $1$ and $2$ respectively.
Finally, NLA is implemented using the ``quantum scissors''-based protocol considered in \cite{xiang2010heralded,winnel2020generalized} as pictured in the ``Noiseless Amplification'' box of Fig.~\ref{fig:bothchannels}. 
Within the amplification step an ancilla photon is introduced into one input of a tunable beamsplitter, one output of which is mixed with the input signal at a separate 50/50 beamsplitter, where a Bell measurement is performed. Successful amplification is heralded by the detection of one and only one photon at either of the detectors (an anti-coincidence event). The gain of each noiseless amplifier is determined by the amplitude transmittance $0\leq t_j \leq 1$
of the tunable beamsplitter,
with ${g_{j}=t_{j}/\sqrt{1-t_{j}^{2}}}$ \cite{winnel2020generalized}.

The resulting density matrix after post-selection 
on the outcomes pictured in Fig.~\ref{fig:bothchannels} and tracing over modes $R$ and $R'$ is given by:
\begin{align}
\hat{\rho}\propto &\left[\nu_1^2(1-\tau_1^2)+\nu_2^2(1-\tau_2^2)\right](1-t_{1}^{2})(1-t_{2}^{2})\vert 00\rangle \langle 00\vert + \nonumber \\
&\left(\nu_1\tau_1 t_{1}\sqrt{1-t_{2}^{2}}\vert 01 \rangle+\nu_2\tau_2\sqrt{1-t_{1}^{2}}t_{2}\vert 10\rangle\right)\left(\nu_1\tau_1t_{1}\sqrt{1-t_{2}^{2}}\langle 01 \vert+\nu_2\tau_2\sqrt{1-t_{1}^{2}}t_{2}\langle 10\vert\right) 
\end{align}
\end{widetext}
assuming all $t_j$, $\nu_{j}$, and $\tau_{j}$ are real and a $\pi/2$ phase shift is added upon reflection from each beamsplitter.

\begin{figure*}[t]
    \centering
    \includegraphics[width=1.8\columnwidth]{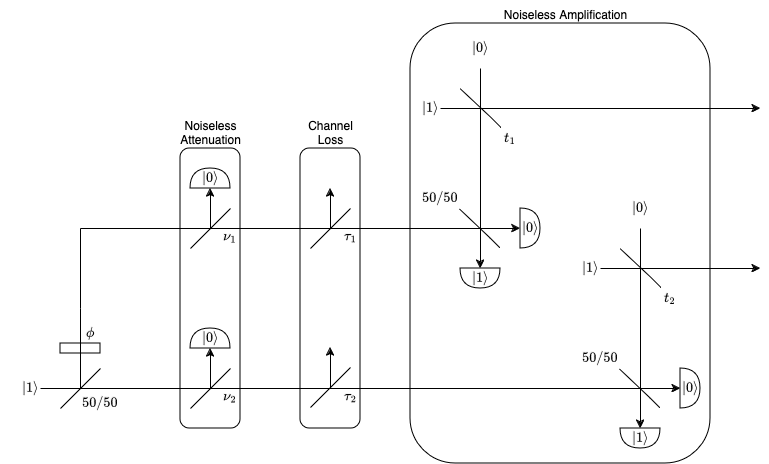}
    \caption{Proposed Optical Realization: This diagram illustrates the setup for the proposed quantum optics experiment. A single photon is initially generated in a balanced superposition of two distinct spatial modes. This superposition is achieved using a 50/50 beam splitter (BS) followed by a $\phi=\pi/2$ phase delay, ensuring the modes' desired phase relationship. The experiment proceeds through several stages as per the Noiseless Loss Suppression (NLS) protocol, which includes noiseless attenuation (NA), simulation of channel loss (L), and noiseless amplification (NLA). Each stage is represented as a distinct box within the figure. The amplitude transmittance of each beam splitter is annotated directly on the respective elements. It is important to note that all beam splitters in this setup are designed to impart no phase shift upon transmission while introducing a phase shift of $i$ ($\pi$/2 radians) on reflection. Detectors are strategically placed to monitor conditional outcomes and are labeled accordingly.}
    \label{fig:bothchannels}
\end{figure*}

Then, in order to recover a result analogous to the single-mode treatment in \cite{mivcuda2012noiseless} where the resultant state can be expressed as a mixture between the initial state and vacuum terms, we require the following condition
\begin{equation}
\begin{aligned}
\sqrt{\mu}=&\sqrt{2}\nu_1\tau_1t_{1}\sqrt{1-t_{2}^{2}}=\sqrt{2}\nu_2\tau_2 t_{2}\sqrt{1-t_{1}^{2}} \\
=&\sqrt{2}\nu_1\tau_1g_{1}\sqrt{\left(1-t_{1}^{2}\right)\left(1-t_{2}^{2}\right)}\\ =&\sqrt{2}\nu_2\tau_2g_{2}\sqrt{\left(1-t_{1}^{2}\right)\left(1-t_{2}^{2}\right)}\\
\end{aligned}
\end{equation}
Note that 
$0\leq \sqrt{\mu} \leq \sqrt{2}$ due to the restrictions on $\nu_{j}$, $\tau_{j}$, and $\phi_{j}$.
For convenience, we define 
\begin{equation}
\sigma=\left[\nu_1^2(1-\tau_1^2)+\nu_2^2(1-\tau_2^2)\right](1-t_{1}^{2})(1-t_{2}^{2})
\end{equation}
so that the state becomes 
\begin{equation}\label{eqn:rhof}
    \hat{\rho}\propto \mu\vert \psi\rangle\langle \psi\vert+ \sigma\vert 00\rangle\langle 00\vert.
\end{equation}

Our definition of $\mu$ can be equivalently expressed as the following restriction on attenuation and gain parameters in each channel 
\begin{equation}\label{eqn:balance}
    \frac{\nu_{1}\tau_{1}g_{1}}{\nu_{2}\tau_{2}g_{2}}=1
\end{equation}
which is analogous to the gain constraint $g=(\nu\tau)^{-1}$ in the single-mode case (see Ref.~\cite{mivcuda2012noiseless} and Sec.~\ref{sec:Micuda}) which recovers the form of the initial state with added vacuum terms. 
Note that applying the same constraint to each channel individually satisfies the two-mode relationship found here with ${\nu_1\tau_1g_1=\nu_2\tau_2g_2=1}$.
More generally, however, the number of free parameters has doubled, so other configurations of gain and attenuation that recover the equivalent output state are possible. It's only important that the overall amplitudes in modes 1 and 2 are balanced. This ``balancing condition'' in Eq.~\ref{eqn:balance} ensures the transmitted state has equal $\ket{01}$ and $\ket{10}$ amplitudes like the input state and is reminiscent of local filtering operations used in Procrustean filtering \cite{kwiat2001experimental,bennett1996concentrating} and nonlocal compensation of polarization-dependent loss \cite{kirby2019effect,jones2018tuning}.

To complete the NLS protocol, we must ensure that the vacuum noise terms can be made arbitrarily small while maintaining the form of Eq.~\ref{eqn:rhof}, e.g., while enforcing the balancing condition given by Eq.~\ref{eqn:balance}.
To this end, we consider the vacuum-signal ratio $\sigma/\mu$ which goes to zero as $\sigma\rightarrow 0$ or $\mu\rightarrow\infty$. The ratio is given by:
\begin{equation}
    \frac{\sigma}{\mu}=\frac{\sigma}{\sqrt{\mu}\sqrt{\mu}}=\frac{\nu_1^2(1-\tau_1^2)+\nu_2^2(1-\tau_2^2)}{2\nu_1\tau_1g_{1}\nu_2\tau_2g_{2}},
    \label{eq:ratio}
\end{equation}
where we have used both equivalent definitions of $\sqrt{\mu}$ in the denominator ($\nu_1\tau_1g_{1}=\nu_2\tau_2g_{2}$ in the resulting expression) as it makes the general behavior of the system more transparent but does not impact the results.
As expected, whenever the channels are lossless ($\tau_{1}=\tau_{2}=1$), the ratio goes to zero.

The first case of interest is when the single-mode NLS protocol is applied to each channel independently, with ${g_j=(\nu_j\tau_j)^{-1}}$.
Inserting this into Eq.~\ref{eq:ratio}, we obtain:
\begin{equation}
    \frac{\sigma}{\mu}\bigg|_{g_{j}=(\nu_{j}\tau_{j})^{-1}}=\frac{\nu_1^2(1-\tau_1^2)+\nu_2^2(1-\tau_2^2)}{2},
    \label{eq:nu_to_zero}
\end{equation}
which, for a fixed channel loss (fixed $\tau$), goes to zero as both $\nu_{j}\rightarrow 0$.
Hence, the original proposal from \cite{mivcuda2012noiseless} works in the two-mode setting when applied to each channel independently. 

\begin{figure}[b]
    \centering
    \includegraphics[width=0.9\columnwidth]{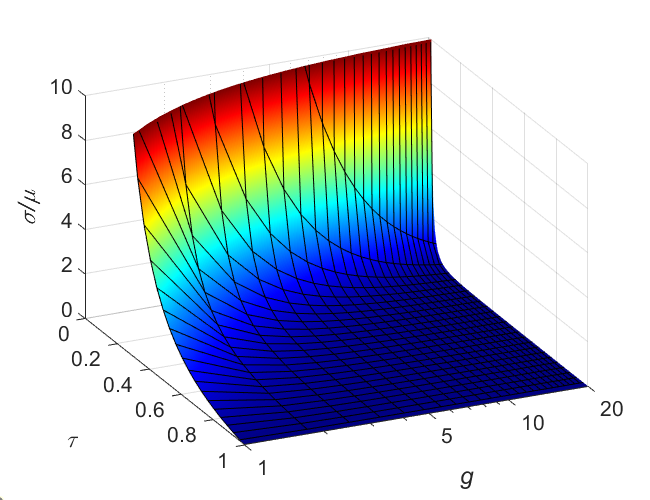}
    \caption{Plot of the vacuum-signal ratio $\sigma/\mu$ as a function of channel transmittance $\tau=\tau_{1}=\tau_{2}$ and NLA gain $g=g_{1}=g_{2}$. 
    Increasing the gain can make the vacuum amplitude arbitrarily small, even with significant channel loss (small $\tau$).}
    \label{fig:plot}
\end{figure}

However, the two-mode setting allows for an additional approach to make the ratio in Eq.~\ref{eq:ratio} arbitrarily small: fixing the $\nu_{j}$ terms and increasing the $g_{j}$.
This can be understood intuitively by setting $\nu_{1}=\nu_{2}$, resulting in
\begin{equation}
    \frac{\sigma}{\mu}\bigg|_{\nu_{1}=\nu_{2}}=\frac{(1-\tau_1^2)+(1-\tau_2^2)}{2\tau_1g_{1}\tau_2g_{2}}.
    \label{eq:nu1}
\end{equation}
In this case, the balancing condition only requires that ${g_1\tau_1=g_2\tau_2}$, allowing arbitrarily high gains in the denominator that drive the ratio $\sigma/\mu$ toward zero.
Furthermore, the $\nu$ terms cancel out entirely, for which $\nu_{j}=1$ is a special case with no noiseless attenuation.
In other words, the extra degrees of freedom afforded by the multi-mode case allow for a version of NLS which only uses noiseless amplification at the end of each channel.
Figure~\ref{fig:plot} shows the vacuum-signal ratio $\sigma/\mu$ for a specific example of this case (Eq.~\ref{eq:nu1}) with $\tau=\tau_{1}=\tau_{2}$ and $g=g_{1}=g_{2}$.
With higher gains, the vacuum amplitude $\sigma$ can be made arbitrarily small relative to the signal. 
Intuitively this occurs because in each term in the state $\ket{01}+\ket{10}$, there is a single photon component that can be amplified relative to the other term, and hence it's possible to re-weight these terms through amplification alone. Contrast this with the single-mode case in section \ref{sec:Micuda} that uses the state $\ket{0}+\ket{1}$; where the inability to change the coefficient of the vacuum term arbitrarily using amplification necessitates the use of noiseless attenuation.

In short, the NLS protocol can be applied as normal in each channel to faithfully transmit the state with low noise. However, the effect of noiseless attenuation is made completely redundant by the amplifiers for this kind of entangled input state. Noiseless attenuation provides no advantage as it only reduces the probability of success, forcing us to increase the gain at the end of each channel to compensate for the reduced amplitude.
This can be seen by considering the probability of success of the NLS protocol, which is taken directly from the form of Eq.~\ref{eqn:rhof} to be $P_s=\text{Tr}\{\mathcal{M}(\hat{\rho}_{in})\}= \mu + \sigma$.
Since both $\mu$ and $\sigma$ decrease as $\nu_{i}$ decreases, the probability of success of the NLS scheme is reduced as more noiseless attenuation is applied.
As we will see in later sections, the redundancy of the noiseless attenuation also holds more generally for M-mode W-states but not for M-mode GHZ states or TMSV states. 
In section \ref{sec:necessity}, we will provide an abstract argument for when this is possible with a general input state. 

\section{W state}
\label{sec:W}

Here we generalize the results of the previous section for one photon in a superposition between two modes to one photon in a superposition of $M$ modes. 
States of this form are known as $W$ states,
and have the property that any two modes remain entangled even when all other modes are discarded \cite{dur2000three}. The nonlocality of single-photon $W$ states has been studied closely \cite{heaney2011extreme}, with applications in long-baseline telescopy \cite{gottesman2012longer} and sensing \cite{guha2013reading}.

An $M$-mode $W$ state with a single-photon can be generated analogously to the two-mode state by replacing the beamsplitter with an $M$-mode multiport splitter, which generates
\begin{equation}
\begin{aligned}
    \vert W\rangle &\propto \vert100...0\rangle+\vert 010...0\rangle+...+\vert000...1\rangle\\
    &=\sum_{j=1}^{M}\vert j\rangle,\\
\end{aligned}
\end{equation}
where we define $\vert j\rangle$ as the state with the single photon in mode $j$ and all other modes unoccupied, and $\vert 0\rangle$ will indicate the vacuum state for all modes. Normalization is omitted and will be included at the end of the calculation.
If each channel now undergoes noiseless attenuation (NA) of magnitude $\nu_{j}$, the state becomes 
\begin{equation}
\begin{aligned}
    \vert W_{\text{NA}}\rangle &\propto\sum_{j=1}^{M}\nu_{j}\vert j\rangle.\\
\end{aligned}
\end{equation}

We now include channel loss by introducing a beamsplitter in each mode with transmittance $\tau_{j}$ where a phase of $i$ is added upon reflection and no phase is added in transmission. We further include the loss modes (the modes where a reflected photon enters) as $\vert j\rangle_{R}$ which indicates a single photon in the reflected port of the $j$th attenuation beam splitter and $\vert 0\rangle_{R}$ which means all reflected modes are in the vacuum states. After splitting all $M$ modes, the $W$ state then becomes 
\begin{equation}
\begin{aligned}
    \vert W_{\text{L}}\rangle&\propto\sum_{j=1}^{M}\tau_{j}\nu_{j}\vert j\rangle\otimes\vert 0\rangle_{R}+\sum_{k=1}^{M}i\sqrt{1-\tau_{k}^{2}}\nu_{k}\vert 0\rangle\otimes\vert k\rangle_{R},\\
\end{aligned}
\end{equation}
where the first sum of terms corresponds to instances where the photons are transmitted, and the second sum corresponds to when the photons are lost (reflected). We have assumed an $i$ phase term for reflection.
We can now include NLA in each channel with gain $g_{j}$, which acting in the $j$th mode takes $\vert j\rangle \rightarrow g_{j}\vert j\rangle$ and $\vert 0\rangle\rightarrow\vert 0\rangle$. The resulting $W$ state is then
\begin{equation}
\begin{aligned}
    \vert W_{\text{NLA}}\rangle&\propto\sum_{j=1}^{M}g_{j}\tau_{j}\nu_{j}\vert j\rangle\otimes\vert 0\rangle_{R}+\sum_{k=1}^{M}i\sqrt{1-\tau_{k}^{2}}\nu_{k}\vert 0\rangle\otimes\vert k\rangle_{R}\\
\end{aligned}
\end{equation}
If we now trace out the loss modes, we have the density matrix given by:
\begin{equation}
        \hat{\rho}\propto\sum_{j,k=1}^{M}g_{j}g_{k}\tau_{j}\tau_{k}\nu_{j}\nu_{k}\vert j\rangle\langle k\vert+\left(\sum_{l=1}^{M}\left(1-\tau_{l}^{2}\right)\nu_{l}^{2}\right)\vert 0\rangle \langle 0\vert.
\end{equation}

As an illustrative case, we consider using the same gain setting as in the single-mode NLS protocol, meaning each ${g_{j}=(\nu_{j}\tau_{j})^{-1}}$. We obtain the final output state
\begin{equation}
    \begin{aligned}
        \hat{\rho}_{out}&\propto\sum_{j,k=1}^{M}\vert j\rangle\langle k\vert+\left(\sum_{l=1}^{M}\left(1-\tau_{l}^{2}\right)\nu_{l}^{2}\right)\vert 0\rangle \langle 0\vert\\
        &=\vert W\rangle\langle W\vert+\frac{1}{M}\left(\sum_{l=1}^{M}\left(1-\tau_{l}^{2}\right)\nu_{l}^{2}\right)\vert 0\rangle \langle 0\vert\\
    \end{aligned}
\end{equation}
where we recover the original result that driving $\nu_{i}\rightarrow 0$ removes the vacuum terms.

More generally, we can recover the original $W$ state form whenever the following condition holds 
\begin{equation}
g_{k}\tau_{k}\nu_{k} = C_{W}, \quad k \in \{1, 2, \ldots, M\}
\label{eq:W_balance_condition}
\end{equation}
where $C_W$ is some arbitrary constant. This is equivalent to the balancing condition in Eq.~\ref{eqn:balance}, in which the product ${g_k\tau_k\nu_k}$ must be equal for two modes, now generalized to $M$ modes.

With this balancing condition enforced, the resulting density matrix is given by
\begin{equation}
        \hat{\rho}\propto C_{W}^{2}\vert W\rangle\langle W\vert+\frac{1}{M}\left(\sum_{l=1}^{M}\left(1-\tau_{l}^{2}\right)\nu_{l}^{2}\right)\vert 0\rangle \langle 0\vert.
\end{equation}
For fixed $\nu_{j}$ and $\tau_{j}$, the vacuum terms can still be made arbitrarily small relative to the $C_{W}^{2}$ coefficient of the $W$ state term by increasing the gain subject to the condition of Eq.~\ref{eq:W_balance_condition}.
The original NLS protocol can be applied to each mode independently, but it is advantageous to do away with the noiseless attenuators ($\nu=1$) and use only noiseless amplifiers. 
Hence, the same behavior for bipartite entanglement in Sec.~\ref{sec:experiment} appears for the more general W state, wherein noiseless amplification alone can faithfully preserve the form of the input state.

\begin{figure}[tb]
    \centering
    \includegraphics[width=\columnwidth]{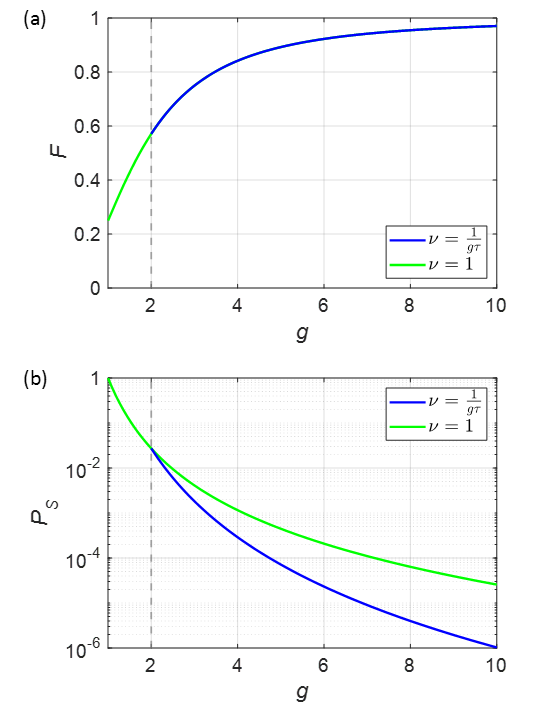}
\caption{(a) W state fidelity $F$ and (b) probability of success $P_s$ after NLS with $(\nu=\frac{1}{g\tau})$ and without $(\nu=1)$ noiseless attenuation, as a function of gain $g$ with all channels equal $\tau_{k}=\tau_{l}$ and $g_{k}=g_{l}$). In each plot $M=3$ and channel transmittance $\tau^{2}=0.25$ for all modes. Note that the regime $g<1/\tau=2$ is unphysical for the case that includes noiseless attenuation, hence the switch from green (light gray) to blue (dark gray) at the vertical dashed line in (a) (where blue (dark gray) and green (light gray) are both overlapping for $g\ge 2$).\label{fig:W_plots}}
\end{figure}

This can be formally seen by considering the fidelity of the output state after NLS
\begin{equation}\label{eqn:W_fidelity}
    F=\frac{C_{W}^{2}}{C_{W}^{2}+\frac{1}{M}\sum_{l=1}^{M}\left(1-\tau_{l}^{2}\right)\nu_{l}^2},
\end{equation}
as well as the probability of success $P_{s}$ given by
\begin{equation}\label{eqn:W_pos}
    P_s =\left(\prod_{k=1}^M g_k^{-2}\right)\left(C_{W}^{2}+\frac{1}{M}\sum_{l=1}^{M}\left(1-\tau_{l}^{2}\right)\nu_{l}^2\right).
\end{equation}
Here we have assumed that noiseless amplification is performed with the filter ${G_N(g)=g^{-N}\sum_{n=0}^{N} g^n \ket{n}\bra{n}}$ described in Sec.~\ref{sec:Micuda}, with $N=1$~\footnote{All probability of success calculations beyond Sec.~\ref{sec:experiment} assume NLA is implemented with the filter $G_N$, with $N$ set to the maximum number of photons in each channel. For the W and GHZ states $(N=1)$, we could also have used the physical amplifiers from in Sec.~\ref{sec:experiment}, which are identical to the $G_N$ filter apart from an additional prefactor. However, the simple $G_N$ filter is much easier to generalize to states with $N>1$ photons, allowing for better comparison with NOON and TMSV states. Crucially, the success probability of the $G_N$ filter exhibits a very similar exponential dependence on amplifier gain that we would also expect from the physical amplifiers in Ref.~\cite{ralph2009nondeterministic}. }. 
These $F$ and $P_{s}$ expressions encapsulate two special cases of interest: the original NLS protocol with $g_{j}=(\nu_{j}\tau_{j})^{-1}$ (hence, $C_{W}=1$), and the absense of noiseless attenuation (all $\nu_{j}=1$ and $g_{j}=C_W/\tau_j$).
In the former, noise terms are suppressed by a factor $\nu_l$ and fidelity approaches one in the limit $\nu_l\rightarrow0$. In the latter, $C_{W}$ can be increased without bound by appropriately adjusting each $g_{k}$, and we see fidelity can still be pushed arbitrarily close to one as $C_{W}\rightarrow\infty$, even without noiseless attenuation.

In fact, as can be seen in Appendix~\ref{app:general_expr}, when Eq. \ref{eqn:W_fidelity} is expressed such that all channels are equal (all $\tau_{k}=\tau$ and $g_{k}=g$), the two special cases with and without noiseless attenuation are exactly equivalent.
This is seen in Fig.~\ref{fig:W_plots}a, where both equivalent expressions for fidelity approach unity as gain increases. However, the probability of success differs when noiseless attenuation is omitted.

In Fig.~\ref{fig:W_plots}b, $P_s$ is plotted as a function of gain for $M=3$ lossy channels with and without noiseless attenuation with all channels equal (all $\tau_{k}=\tau_{l}$ and $g_{k}=g_{l}$). As seen in the two-mode case in Sec.~\ref{sec:experiment}, noiseless attenuation of $W$ states drastically reduces the probability of success while providing no advantage over noiseless amplification alone.

\section{GHZ State}
\label{sec:GHZ}

To this point, we have only considered states where a single-photon is present in a superposition over some number of modes. Here we will generalize further to states of the GHZ form which are a superposition of $M$ photons, one in each of $M$ modes, with the $M$ mode vacuum. GHZ states can be used to develop a quantum network of clocks \cite{komar2014quantum} and have been used for room-scale demonstrations with superconducting systems for modular quantum computing \cite{zhong2021deterministic}. The GHZ states represent a type of entanglement at odds with W-state entanglement. In particular, the GHZ states can be reduced to a completely separable state through measurement on a single qubit in the computational basis and cannot be transformed into a W state (or the reverse, W to GHZ) through local operations assisted with classical communication (LOCC) \cite{dur2000three}. In an intuitive sense, the $N$-mode GHZ state is actually more similar to the single-photon state used in the original NLS protocol \cite{mivcuda2012noiseless}, as all information is encoded in a superposition of presence or complete absence of a photon in all modes, reminiscent of the $\vert 0\rangle+\vert 1\rangle$ state in a single mode. As we will see, the GHZ states, in contrast to the W-states of the prior section, do in fact require the noiseless attenuation step in order to perform NLS.

We begin by defining the $M$-mode GHZ state:
\begin{equation}
    \begin{aligned}
        \vert \text{GHZ}\rangle &= \frac{1}{\sqrt{2}}\left(\vert 000...0\rangle+\vert 111...1\rangle\right). 
    \end{aligned}
\end{equation}
The first step of the NLS protocol is to apply noiseless attenuation of magnitude $\nu_{j}$ to each mode independently, resulting in
\begin{equation}
    \begin{aligned}
        \vert \text{GHZ}_{\text{NA}}\rangle &\propto\vert 000...0\rangle+\nu_{1}\nu_{2}\nu_{3}...\nu_{M}\vert 111...1\rangle\\
        &=\vert 000...0\rangle+\prod_{k=1}^{M}\nu_{k}\vert 111...1\rangle.\\
    \end{aligned}
\end{equation}

We now include channel loss by including a beamsplitter in each mode with transmittance $\tau_{j}$ where $j$ labels the mode and a phase of $i$ is added on reflection:
\begin{widetext}
\begin{equation}
    \begin{aligned}
        \vert \text{GHZ}_{\text{L}}\rangle &\propto\vert 000...0\rangle\otimes\vert 000...0\rangle_{R}+\left(\sum_{\mathbf{j}\in B^{M}}\left[\prod_{k=1}^{M}\left(\tau_{k}\right)^{j_{k}}\left(i\sqrt{1-\tau_{k}^{2}}\right)^{\neg j_{k}}\nu_{k}\right]\vert j_{1} j_{2} j_{3}...j_{M}\rangle\otimes\vert \neg\left(j_{1} j_{2} j_{3}...j_{M}\right)\rangle_{R} \right).\\ 
    \end{aligned}
\end{equation}
Here, $\mathbf{j}$ denotes an $M$-tuple representing a binary string of length $M$. The set of all such $M$-tuples, corresponding to all possible binary strings of length $M$, is represented as $B^M$. Each M-tuple can be written as $(j_{1}, j_{2}, j_{3},..., j_{M})$, where the subscript indicates the position of a bit in the binary string. 
We use the logical $\neg$ symbol to invert a bit string (e.g., $\neg 010=101$), have included the additional beamsplitter loss (reflection) modes as $\vert \cdot\rangle_{R}$, and include a phase shift of $i$ upon reflection.

Finally, we can now include noiseless attenuation in each channel with gain $g_{j}$, which acting in the $j$th mode takes $\vert j\rangle \rightarrow g_{j}\vert j\rangle$ and $\vert 0\rangle\rightarrow\vert 0\rangle$. 
The resulting GHZ state is then 
\begin{equation}
    \begin{aligned}
        \vert \text{GHZ}_{\text{NLA}}\rangle &\propto\vert 000...0\rangle\otimes\vert 000...0\rangle_{R}\\
        &+\left(\sum_{\mathbf{j}\in B^{M}}\left[\prod_{k=1}^{M}\left(g_{k}\right)^{j_{k}}\left(\tau_{k}\right)^{j_{k}}\left(i\sqrt{1-\tau_{k}^{2}}\right)^{\neg j_{k}}\nu_{k}\right]\vert j_{1} j_{2} j_{3}...j_{M}\rangle\otimes\vert\neg\left(j_{1} j_{2} j_{3}...j_{M}\right)\rangle_{R} \right). \\
    \end{aligned}
\end{equation}
For ease of calculation we rewrite this as

\begin{equation}
    \begin{aligned}
        \vert \text{GHZ}_{\text{NLA}}\rangle &\propto\left(\vert 000...0\rangle+\left[\prod_{k=1}^{M}g_{k}\tau_{k}\nu_{k}\right]\vert 111...1\rangle\right)\otimes\vert 000...0\rangle_{R}\\
        &+\left(\sum_{\mathbf{j} \in B^M \setminus 1^M}\left[\prod_{k=1}^{M}\left(g_{k}\right)^{j_{k}}\left(\tau_{k}\right)^{j_{k}}\left(i\sqrt{1-\tau_{k}^{2}}\right)^{\neg j_{k}}\nu_{k}\right]\vert j_{1} j_{2} j_{3}...j_{M}\rangle\otimes\vert\neg\left(j_{1} j_{2} j_{3}...j_{M}\right)\rangle_{R} \right). \\
    \end{aligned}
\end{equation}
where $B^M \setminus {1^M}$ is the set of all possible binary strings of length $M$ excluding the one where all bits are 1. We can now trace out the reflected modes to obtain:
\begin{equation}
    \begin{aligned}
        \hat{\rho}&\propto \left(\ket{000...0} + \left[\prod_{k=1}^{M}g_{k}\tau_{k}\nu_{k}\right]\ket{111...1} \right)\left(\bra{000...0} + \left[\prod_{k=1}^{M}g_{k}\tau_{k}\nu_{k}\right]\bra{111...1} \right)\\
        &+\left(\sum_{\mathbf{j}\in B^M \setminus {1^M}}\left[\prod_{k=1}^{M}\left(g_{k}^{2}\right)^{j_{k}}\left(\tau_{k}^{2}\right)^{j_{k}}\left(1-\tau_{k}^{2}\right)^{\neg j_{k}}\nu_{k}^{2}\right]\vert j_{1} j_{2} j_{3}...j_{M}\rangle\langle j_{1} j_{2} j_{3}...j_{M}\vert \right)
    \end{aligned}
\end{equation}

We see that whenever the condition 
\begin{equation}\label{eqn:ghz_balance}
\prod_{k=1}^{M}g_{k}\tau_{k}\nu_{k}=1
\end{equation}
is met, we recover 
\begin{equation}
    \begin{aligned}
        \hat{\rho}\propto \vert \text{GHZ}\rangle\langle \text{GHZ}\vert +\frac{1}{2}\left(\sum_{\mathbf{j}\in B^M \setminus {1^M}}\left[\prod_{k=1}^{M}\left(g_{k}^{2}\right)^{j_{k}}\left(\tau_{k}^{2}\right)^{j_{k}}\left(1-\tau_{k}^{2}\right)^{\neg j_{k}}\nu_{k}^{2}\right]\vert j_{1} j_{2} j_{3}...j_{M}\rangle\langle j_{1} j_{2} j_{3}...j_{M}\vert \right).
    \end{aligned}
\end{equation}
This condition is akin to that of the single-mode case, which this reduces to for $k=1$, becoming $g\tau\nu=1$, the same as found in \cite{mivcuda2012noiseless} and in section \ref{sec:Micuda}.

If we perform single-mode NLS in each channel with $g_k=(\nu_k\tau_k)^{-1}$, the fidelity of the output state is given by 
\begin{align}
    F=&\frac{1+\frac{1}{4}\prod_{k=1}^{M}(1-\tau_k^2)\nu_k^2}{1+\frac{1}{2}\left(\sum_{\mathbf{j}\in B^M \setminus {1^M}}\left[\prod_{k=1}^{M}\left((1-\tau_{k}^{2})\nu_k^2\right)^{\neg j_{k}}\right]\right)}
\end{align}
with a probability of success:
\begin{align}
    P_s =&  \left(\prod_{k=1}^M g_k^{-2}\right)
    \left[1+\frac{1}{2}\left(\sum_{\mathbf{j}\in B^M \setminus {1^M}}\left[\prod_{k=1}^{M}\left((1-\tau_{k}^{2})\nu_k^2\right)^{\neg j_{k}}\right]\right)\right]
\end{align}
Following our analysis of the $W$ state, we can explore the effect of noiseless amplification alone on GHZ states by setting all $\nu_j=1$ and introducing an arbitrary constant $g_j\tau_j=C_{GHZ}$. In this case, the output state takes a similar form:
\begin{equation}
    \begin{aligned}
        \hat{\rho}&\propto \frac{1}{2}\left(\ket{000...0} + C_{GHZ}^M\ket{111...1} \right)\left(\bra{000...0} + C_{GHZ}^M\bra{111...1} \right)\\
        &+ \frac{1}{2}\left(\sum_{\mathbf{j}\in B^M \setminus {1^M}}\left[\prod_{k=1}^{M}\left(C_{GHZ}^2\right)^{j_{k}}\left(1-\tau_{k}^{2}\right)^{\neg j_{k}}\right]\vert j_{1} j_{2} j_{3}...j_{M}\rangle\langle j_{1} j_{2} j_{3}...j_{M}\vert \right)
    \end{aligned}
\end{equation}
\end{widetext}
As seen before, driving the gain higher in the limit ${C_{GHZ}\rightarrow\infty}$ will suppress the noise terms generated by photon loss. However, the balancing condition in Eq.~\ref{eqn:ghz_balance} is only met for $C_{GHZ}=1$, so the input state is no longer preserved at higher gain values. General expressions for fidelity and probability of success are given in Appendix~\ref{app:general_expr}.

\begin{figure}[tb]
    \centering
    \includegraphics[width=\columnwidth]{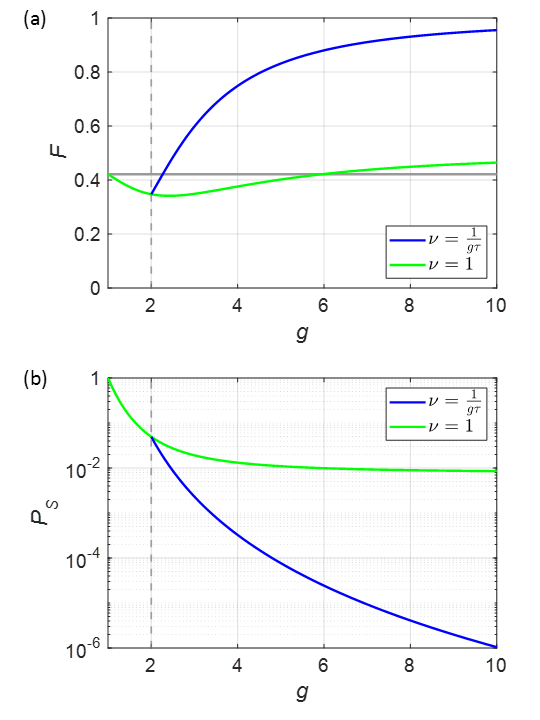}
\caption{(a) GHZ state fidelity $F$ and (b) probability of success $P_s$ after NLS with $(\nu=\frac{1}{g\tau})$ and without $(\nu=1)$ noiseless attenuation, as a function of gain. Channel transmittance $\tau^2=0.25$ for all $M=3$ modes. Note that the regime $g<1/\tau=2$ is unphysical for the case that includes noiseless attenuation. The original output fidelity for pure loss without noiseless attenuation or amplification $(\nu=g=1)$ is indicated by the solid gray horizontal line.\label{fig:GHZ_plots}}
\end{figure}

Each of these cases, with and without noiseless attenuation, are compared in Fig.~\ref{fig:GHZ_plots} for $M=3$ modes with identical loss. Ordinary NLS with $g_k=(\nu_k\tau_k)^{-1}$ allows for fidelities arbitrarily close to one as gain is increased, at the cost of vanishing probability of success. In contrast, direct amplification without noiseless attenuation reduces fidelity for low values of gain $g\lesssim 6$ and only ever reaches $F=\frac{1}{2}$ in the limit of infinite gain. In that limit, the amplified output state is simply the product state $\ket{111}$, so all entanglement has been lost. 
Notably, the probability of success tends to a nonzero value at infinite gain $P_s\rightarrow \tau^{2M}/2$. This is also what we would expect using physical single-photon amplifiers like that in Section~\ref{sec:experiment}. When the gain is infinite, the incoming signal photons are effectively just replaced with unentangled ancilla photons, and probability of success is determined entirely by the transmission of the signal photons across all lossy channels $\tau^{2M}$, which heralds successful amplification once they are detected.

To this point, we have only considered states with at most one photon per mode. In the next two sections, we relax this condition and study states with multiple photons per mode and assess the impact of NLS. 

\section{Two-Mode Squeezed Vacuum}
\label{sec:TMSV}

Entanglement in quantum optics experiments is commonly generated via a nonlinear process that results in a two-mode squeezed vacuum state (TMSV). 
These states are given by:
\begin{equation}
    \ket{\text{TMSV}}=\sqrt{1-\gamma}\sum_{n=0}^\infty{\gamma^n}\ket{n}_{1}\ket{n}_{2}.
\end{equation}
States of this form are often considered in the context of entanglement distillation with noiseless amplifiers, which is central to many continuous variable quantum communication protocols~\cite{ralph2009nondeterministic,ulanov2015undoing,seshadreesan2019continuous,liu2022distillation}. Here we also consider noiseless amplification of TMSV states as a benchmark against a full NLS protocol (noiseless attenuation included), comparing their utility for high-fidelity transmission across a lossy channel.

In order to apply the NLS protocol to this state, we first attenuate modes $1$ and $2$ by $\nu_{1}$ and $\nu_{2}$, resulting in another TMSV with a smaller amplitude 
\begin{equation}
\ket{\text{TMSV}_{\text{NA}}}\propto\sum_{n=0}^\infty{\gamma^n}(\nu_1\nu_2)^n\ket{n}_{1}\ket{n}_{2}.
\end{equation}
Photon loss is then applied using beamsplitters with transmission coefficients $\tau_{1}$ and $\tau_{2}$ and then traced over resulting in a mixed state of the form:
\begin{equation}\label{eqn:rho_tmsv_full}
    \hat{\rho}\propto \sum_{j,k=0}^{\infty}\ket{\psi_{jk}}\bra{\psi_{jk}}
\end{equation}
where the loss of $j$ photons from mode $1$ and $k$ photons from mode $2$ corresponds to the following (unnormalized) state:
\begin{equation}\label{eqn:psi_tmsv}
    \ket{\psi_{jk}}=\sum_{n\geq j,k} c_{njk}\ket{n-j}_{1}\ket{n-k}_{2}
\end{equation}
with the coefficients:
\begin{equation}
\begin{aligned}
    &c_{njk}\equiv \\
    &\sqrt{\binom{n}{j}\binom{n}{k}} (\gamma \nu_1 \nu_2)^n (\tau_1)^{n-j} (\tau_2)^{n-k} (i\sqrt{1-\tau_1^2})^j (i\sqrt{1-\tau_2^2})^{k}. \\
    \end{aligned}
\end{equation}
After noiseless amplification with the filters $G_{N}(g_l)$ in each mode $l=1,2$, we modify these coefficients with two extra gain factors and introduce a maximum number of photons $N$ in each mode:
\begin{equation}\label{eqn:cnjk_prime}
    c'_{njk} \equiv \begin{cases}
        c_{njk} g_1^{n-j} g_2^{n-k} & \text{if } n-j \leq{N} \text{ and } n-k \leq N \\
        0 & \text{otherwise}
    \end{cases}
\end{equation}
The first term of the output state $\hat{\rho}$ corresponds to the case of no photon loss $(j=k=0)$. The Fock state coefficients of this lossless term are:
\begin{equation}
    c'_{n00}=(\gamma\nu_1\nu_2\tau_1\tau_{2}g_{1}g_{2})^n
\end{equation}
which corresponds to a pure TMSV state. To ensure these coefficients are the same as for the initial state, we arrive at the following condition:
\begin{equation}\label{eqn:TMSV_balance}
    (\nu_1\tau_1 g_1)(\nu_2\tau_2 g_2) = 1,
\end{equation}
which is equivalent to the GHZ state balancing condition in the previous section. Satisfying this condition, the output once again consists of the original state plus added noise. 
As expected, if we implement NLS by setting $g_1=1/(\nu_1\tau_1)$ and $g_2=1/(\nu_2\tau_2)$ and consider the limits $\nu_{1,2}\rightarrow0$, the coefficients scale as
\begin{equation}\label{eqn:TMSV_scaling}
    c^{(NLS)}_{njk}\propto \nu_1^j\nu_2^k, 
\end{equation}
and all terms with $j>0$ ($k>0$) will disappear as $\nu_1\rightarrow 0$ ($\nu_2\rightarrow0$).
Hence, we conclude that the TMSV behaves similarly to the GHZ states in that the application of NLS requires noiseless attenuation and the mixing due to photon loss can be made arbitrarily small.

The fidelity for a general output state can be calculated as follows:
\begin{align}
    F&=\frac{\sum_{j,k}|\langle \psi | \psi'_{jk}\rangle|^2}{\sum_{j,k}|\psi'_{jk}|^2} \\
    &=\frac{(1-\gamma^2)\sum_{j=0}^\infty\left(\sum_{n=j}^{N+j}|c'_{njj}|\gamma^{(n-j)}\right)^2}{\sum_{j,k=0}^\infty\sum_{n\geq j,k}|c'_{njk}|^2} \label{eqn:TMSV_fidelity}  
\end{align}
where $\ket{\psi'_{jk}}$ are defined as in Eq.~\ref{eqn:psi_tmsv} with the modified coefficients $c'_{njk}$ given in Eq.~\ref{eqn:cnjk_prime}. The probability of success is given by:
\begin{align}
    P_s&=(g_1g_2)^{-2N}(1-\gamma^2)\sum_{j,k}|\psi'_{jk}|^2.
\end{align}
Following Eq.~\ref{eqn:TMSV_scaling}, all but the $j=k=0$ terms in the above expressions will vanish in the limit of $\nu_{1,2}\rightarrow 0$, leaving only sums over the coefficients $c'_{n00}=\gamma^n$ with ${g_l=(\nu_l\tau_l)^{-1}}$. Evaluating Eq.~\ref{eqn:TMSV_fidelity} in this limit confirms that fidelity approaches unity for sufficiently large $N$.

Similar to the GHZ state, we compare the fidelity and probability of success for two special cases. The first is the typical NLS protocol with each ${g_l=(\nu_l\tau_l)^{-1}}$, which satisfies the balancing condition in Eq.~\ref{eqn:TMSV_balance}. The second case uses no noiseless attenuation $(\nu_l=1)$ and defines the arbitrary constant $C_{TMSV}=g_1\tau_1=g_2\tau_2$, which does increase state purity as $C_{TMSV}\rightarrow\infty$ but does not satisfy the same balancing condition.

\begin{figure}[ht]
    \centering
    \includegraphics[width=\columnwidth]{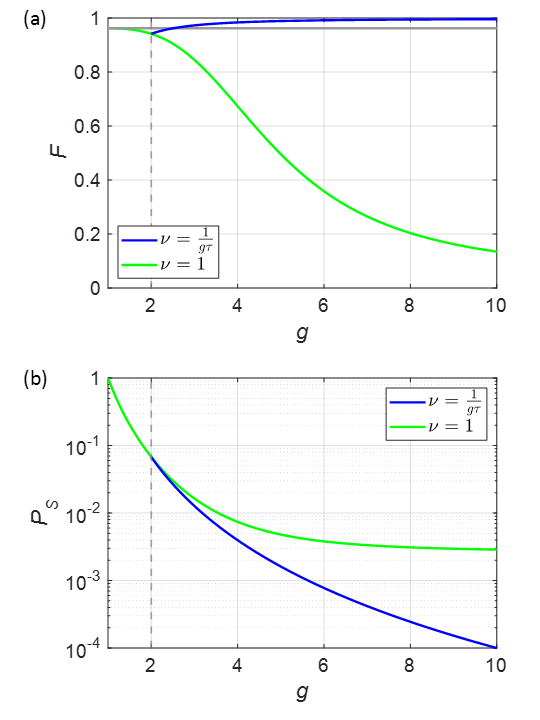}
\caption{(a) TMSV state fidelity $F$ and (b) probability of success $P_s$ after NLS with $(\nu=\frac{1}{g\tau})$ and without $(\nu=1)$ noiseless attenuation, as a function of gain. TMSV amplitude is $\gamma=0.2$, truncated at $N=1$. Channel transmittance $\tau^2=0.25$ for all modes. Note that the regime $g<1/\tau=2$ is unphysical for the case that includes noiseless attenuation.}\label{fig:TMSV_plots}
\end{figure}

The main results are shown in Fig.~\ref{fig:TMSV_plots}. The output fidelity and probability of success are plotted for a TMSV input with amplitude $\gamma=0.2$, equal channel losses of $\tau_1=\tau_2=0.5$, and a maximum photon number $N=1$ for noiseless amplification. Because of this truncation, the output state will lack higher-order terms and thus fidelity can never reach unity in the limit of high gain. This can be overcome to an extent by increasing $N$, but at the expense of probability of success since $P_s\propto g^{-2N}$. Even so, Fig.~\ref{fig:TMSV_plots}a shows very high fidelity can be achieved when retaining only single-photon terms, with $F\rightarrow 1-\gamma^4$ in the limit of infinite gain. Note that the increase in fidelity appears quite modest, because the initial value is already quite high without applying any noiseless attenuation or amplification $(\nu=g=1)$. This is owed to the large vacuum term in the initial TMSV, ensuring large overlap with a state sent through a pure loss channel. Importantly, though, NLS restores the purity of the original entangled state by extinguishing the noise terms in the output state mixture.

This is in stark contrast to the case where noiseless attenuation is omitted, which despite the marginally improved probability of success, shows fidelity decreasing to a small finite value, reaching $F\rightarrow\gamma^2-\gamma^4$ in the limit of infinite gain with $N=1$. Noiseless amplification on its own weights the highest-number Fock states above all others, so $\rho\rightarrow\ket{N}\bra{N}$ in the limit of high gain. Purity increases, but all entanglement is lost in this limit, and fidelity with the original TMSV is minimal. As shown by Ulanov et al.~\cite{ulanov2015undoing}, noiseless amplification with certain intermediate values of gain can still distill entanglement beyond that of the original TMSV state. In that scenario, the purity of the distilled output states is ensured by the small squeezing term of the initial state, $\gamma\ll 1$.  In that case, the small $\gamma$ plays an analogous role to $\nu$ of noiseless attenuation in the NLS approach since it suppresses the relative weight of the non-vacuum terms prior to transmission through the lossy channel.
Although we have chosen the parameters $\nu_k=1/(g\tau_k)$ to maximize fidelity with the initial state, noiseless attenuation paired with the optimal gain values of Ref.~\cite{ulanov2015undoing} can instead be chosen to maximize entanglement.

\section{NOON states}
\label{sec:noon}

The NOON states are a type of bipartite entangled state of $N$ photons where all are either in one mode or the other. 
These states are of particular interest in quantum metrology due to the scaling of interference effects, and hence phase sensitivity, with photon number~\cite{dowling2008quantum,lang2014optimal}. NOON states have also recently been experimentally demonstrated with frequency bin states which could ultimately be useful for high-dimensional quantum key distribution (QKD) \cite{lee2024noon}.
Beyond practical applications, NOON states are also of fundamental interest since, in the limit of large $N$, the entanglement is between macroscopic systems.

The NOON state is given by: 
\begin{equation}
    \ket{\psi}=\frac{1}{\sqrt{2}}\left(\ket{n}_1\ket{0}_2+\ket{0}_1\ket{n}_2\right)
\end{equation}
To apply NLS, we first separately apply noiseless attenuation to each mode with attenuation coefficients $\nu_{1}$ and $\nu_{2}$, resulting in
\begin{equation}
    \ket{\psi_{NA}}\propto \nu_1^{n}\ket{n}_1\ket{0}_2 + \nu_2^n\ket{0}_1\ket{n}_2.
\end{equation}
Photon loss in each mode is now included through the use of beamsplitters with amplitude transmittances $\tau_{1,2}$, resulting in
\begin{equation}
\begin{aligned}
        &\ket{\psi_{L}}\propto\\
        & \nu_1^{n}\sum_{j=0}^n \sqrt{\binom{n}{j}}\tau_1^{n-j}(i\sqrt{1-\tau_1^2})^j\ket{n-j}_1\ket{j}_{R}\ket{0}_2\ket{0}_{R'} \\
        &+ \nu_2^{n}\sum_{j=0}^n \sqrt{\binom{n}{j}}\tau_2^{n-j}(i\sqrt{1-\tau_2^2})^j\ket{0}_1\ket{0}_{R}\ket{n-j}_2\ket{j}_{R'}.\\
\end{aligned}
\end{equation}
where $R$ and $R'$ are the output ports of the photon-loss beamsplitters in modes $1$ and $2$, respectively. 
We now complete the NLS protocol by tracing out the $a$ and $b$ modes and applying noiseless amplification, resulting in

\begin{equation}\begin{aligned}
    &\hat{\rho}_{out} \propto\\
    &\left( \nu_1^{n}\tau_1^n g_1^n\ket{n}_1\ket{0}_2 + \nu_2^{n}\tau_2^n g_2^n\ket{0}_1\ket{n}_2 \right)\left( c.c. \right) \\
    &+ \nu_1^{2n}\sum_{j=1}^n \binom{n}{j}(\tau_1g_1)^{2(n-j)}(1-\tau_1^2)^j\ket{n-j}_1\ket{0}_2\bra{n-j}_1\bra{0}_2 \\
    &+ \nu_2^{2n}\sum_{j=1}^n \binom{n}{j}(\tau_2g_2)^{2(n-j)}(1-\tau_2^2)^j\ket{0}_1\ket{n-j}_2\bra{0}_1\bra{n-j}_2.\\
\end{aligned}
\end{equation}
The first term, corresponding to zero photon loss, will be proportional to the original NOON state if we satisfy the familiar balancing condition $\nu_1\tau_1g_1=\nu_2\tau_2g_2=C_{N}$, analogous to the $W$ state case (see Eq.~\ref{eq:W_balance_condition}). 

The fidelity of the output state is given by:
\begin{equation}
    F=\frac{C_{N}^{2n}}{C_{N}^{2n}+\frac{1}{2}\sum_{l=1}^{2}\sum_{j=1}^{n}\binom{n}{j}C_{N}^{2(n-j)}\left[(1-\tau_{l}^2)\nu_{j}^{2}\right]^j}
\end{equation}
with the corresponding probability of success:
\begin{equation}
    P_s =  \left(\frac{1}{g_1g_2}\right)^{2n}\left(C_{N}^{2n}+\frac{1}{2}\sum_{l=1}^{2}\sum_{j=1}^{n}\binom{n}{j}C_{N}^{2(n-j)}\left[(1-\tau_{l}^2)\nu_{j}^{2}\right]^j\right).
\end{equation}

Once again, we can suppress losses in multiple ways. We can set $g_k=1/(\nu_k\tau_k)$ for each mode as in the original NLS protocol $(C_N=1)$, and each noise term will be weighted by a factor $\nu_k^{2j}$. Alternatively, noiseless attenuation can be omitted with all $\nu_k=1$, and noise terms will still be suppressed by a factor $C_N^{-2j}$ relative to the pure NOON state term. Fig.~\ref{fig:NOON_plots}(a) shows that when plotted as a function of gain, both fidelity expressions approach unity in the high-gain limit, and in fact are equivalent. However, noiseless attenuation is clearly detrimental since high-fidelity transmission will be performed with a much lower probability of success, as seen in Fig.~\ref{fig:NOON_plots}(b). A more explicit comparison to the similarly-behaved $W$ state can be seen in Appendix~\ref{app:all_states}.

\begin{figure}[t]
    \centering
    \includegraphics[width=\columnwidth]{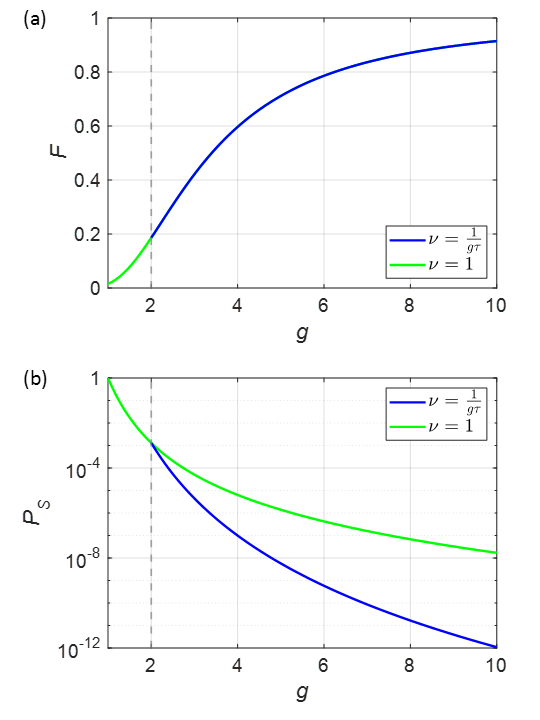}
\caption{(a) $N=3$ NOON state fidelity $F$ and (b) probability of success $P_s$ after NLS with $(\nu=\frac{1}{g\tau})$ and without $(\nu=1)$ noiseless attenuation, as a function of gain. Channel transmittance $\tau^2=0.25$ for all modes. Note that the regime $g<1/\tau=2$ is unphysical for the case that includes noiseless attenuation.}\label{fig:NOON_plots}
\end{figure}

\section{On the necessity of Noiseless Attenuation}
\label{sec:necessity}

As we have seen, the naive application of the original NLS protocol to each mode of an M-mode entangled state independently will result in the desired end state. However, for certain inputs like W and NOON states, as discussed in sections \ref{sec:W} and \ref{sec:noon}, noiseless attenuation is redundant and only reduces the overall probability of success. Here we provide a criteria for when states do or do not require the noiseless attenuation step of NLS.

The NLS protocol has two important and related features that we must ensure remain in the absence of noiseless attenuation: the asymptotic suppression of noise through the arbitrary increase of gain, and the preservation of the input state (achieved in the original protocol through the application of noiseless attenuation proportional to the gain).
The results of this section will show that the first property can be achieved through gain alone, regardless of whether noiseless attenuation is included, but that the latter property, the preservation of the input state, will often require it. 
For example, we find this is true of states that encode information in the vacuum, meaning one term in the superposition includes no photons. The inability to amplify this vacuum term means that noiseless attenuation is required to balance all coefficients and recover the initial state.

In Section \ref{sec:Micuda}, specifically Eq. \ref{eqn:NLStot}, we showed that for the original NLS scheme these two features are most apparent in the natural splitting of the output state into a ``signal'' corresponding to the initial input state we want to preserve and a ``noise'' term containing everything we want to suppress. We then chose a specific relationship for the values of attenuation and gain that simultaneously allows the gain to increase arbitrarily to suppress noise while maximizing signal fidelity for the most general set of input states.
However, NLS admits other configurations of attenuation and amplification. To see this, we can inspect the form of each operator $\hat{B}^{(k)}_{j_k}$ in the Fock basis directly (reprinted from Section \ref{sec:Micuda} for ease of reference): 
\begin{equation}\label{eqn:bk_ops2}
    \hat{B}^{(k)}_{j_k} = \sum_{m=j_k}^{N}\sqrt{\binom{m}{j_k}}(1-\tau_k^2)^{j_k/2}(\nu_k \tau_k g_k )^{m-j_k}\ket{m-j_k}\bra{m}.
\end{equation}
For fixed values of $\nu_k$ and $\tau_k$, the above operators scale with gain as $\hat{B}^{(k)}_{j_k}\sim g_{k}^{N-j_k}$. It follows that the corresponding noise terms $\tilde{\rho}_{j_1\dots j_N}$ which have some $j_k > 0$ (see Eq.~\ref{eqn:NLStot}) can be made arbitrarily small relative to the signal that has all $j_k=0$, disappearing in the limit $g_k\rightarrow\infty$ for all $k$. Previously, this scaling with gain was offset by noiseless attenuation due to the condition ${\nu_k \tau_k g_k = 1}$, and noise suppression was accomplished by tuning the additional attenuation factors present in Eq.~\ref{eqn:NLStot}.
However, the above argument makes it clear that asymptotic suppression of noise is possible through gain alone, whether or not noiseless attenuation is included.

Given that we have now shown the noise suppression is obtainable independent of the inclusion of noiseless attenuation, we turn our attention to the form of the signal term and determine when noiseless attenuation is necessary to preserve its form. 
We start with a completely general pure $M$-mode superposition state:
\begin{equation}\label{eqn:psi_gen}
    \ket{\psi}=\sum_{l=1}^{L} c_l \ket{n_1 \cdots n_M}^{(l)}
\end{equation}
where each of the $L$ terms in the superposition has a unique set of occupation numbers $\{n_k^{(l)}\}$ across $M$ modes.
Next we apply the NLS channel (with no assumptions about the values of $\nu_k$ or $g_k$ yet). Since we already discussed the noise terms, we isolate the signal term ($\tilde{\rho}_{0\cdots 0}$ in Eq. \ref{eqn:NLStot}) in which no photons were lost ${(j_k=0)}$.
This corresponds to a pure state given by the following expression:
\begin{equation}
\label{eqn:no_loss_state} 
\begin{aligned}
     \ket{\tilde{\psi}} &= \left(\prod_{k=1}^M{(\nu_kg_k\tau_k)^{\hat{n}}}\right) \ket{\psi}\\
     &\equiv \hat{X}\hat{Y} \ket{\psi}\\
\end{aligned}
\end{equation}
where, for convenience, we have also defined the \textit{amplifier}-only operator ${\hat{X}\equiv\prod_{k=1}^M{(g_k\tau_k)^{\hat{n}}}}$ and the \textit{attenuator}-only operator ${\hat{Y}\equiv\prod_{k=1}^M{(\nu_k)^{\hat{n}}}}$. Since our focus below will be on the role of attenuation, we have opted, without loss of generality, to include channel loss in $\hat{X}$.
Note that $\hat{X}$ and $\hat{Y}$ commute (or, more generally, all three of the component operators $\nu_k^{\hat{n}}$, $\tau_k^{\hat{n}}$, and $\g_k^{\hat{n}}$) because they are simultaneously diagonal in the Fock basis (but, importantly, do not necessarily have distinct eigenvalues \footnote{A well known theorem in linear algebra states that if two operators commute and each of them has distinct eigenvalues the operators then share the same eigenvectors. Since this is not the case we have (eigenvalues can repeat here) we cannot guarantee from this theorem alone that, for example, an eigenvector of $\hat{X}$ is an eigenvector of $\hat{Y}$}).
Hence, for the following analysis of the signal term alone the time ordering of the operators is unimportant.

For NLS to work as intended, we require $\ket{\tilde{\psi}}\propto\ket{\psi}$, meaning the input state $\ket{\psi}$ is an \textit{eigenstate} of the 
``lossless transmission operator'' $\hat{X}\hat{Y}=\prod_{k=1}^M{(\nu_kg_k\tau_k)^{\hat{n}}}$. 
As described in section \ref{sec:Micuda}, enforcing the condition ${\nu_k\tau_kg_k=1}$ sets this operator to the identity, and the asymptotic gain required to suppress noise is balanced with proportional noiseless attenuation.
Here, we instead consider the NLS protocol without noiseless attenuation $(\nu_k=1)$ and hence seek a solution that maintains amplifier gain as a free parameter (as it needs to be asymptotically large to remove the noise terms) without the ability to offset this with noiseless attenuation. To make this final point more explicit, we note that the original assumption of $\nu_{k}\tau_{k}g_{k}=1$ can be satisfied without noiseless attenuation by simply setting $\tau_{k}g_{k}=1$, which would still preserve the signal. However, since channel loss $\tau_{k}$ is fixed, we would no longer be able to increase gain arbitrarily to suppress noise.

Without noiseless attenuation, we have $\hat{Y}=\hat{I}$ and seek eigenstates of the \textit{amplifier}-only operator $\hat{X}\ket{\psi}\propto \ket{\psi}$. The input is an eigenstate of $\hat{X}$ if and only if each term in the superposition $l=1,\dots,L$ produces the same constant eigenvalue $\lambda$:
\begin{equation}
\prod_{k=1}^M(g_{k}\tau_{k})^{n_{k}^{(l)}} = \lambda
\label{eqn:eigenproblem1}
\end{equation}
This condition can be equivalently expressed as a set of $L$ linear equations:
\begin{equation}
\sum_{k=1}^M n_{k}^{(l)} x_k = \log\lambda 
\label{eqn:eigenproblem2}
\end{equation}
with $x_k\equiv \log(g_k\tau_k)$. A set of amplifiers $g_k$ that satisfy the above equations will preserve the form of the input state. Additionally, to suppress noise terms in the full output state, we also require that these solutions allow for arbitrarily high gain values. With some constant $C>1$, we can construct a new set of solutions $x'_k = C x_k$ (or expressed in terms of gain, $g'_k = g_k^{C}\tau_k^{C-1}$), and this will return a new eigenvalue $\lambda' = \lambda^C$. This constant $C$ can be made arbitrarily large, but we find that this corresponds to increasing gain $g'_k>g_k$ only when $g_k > 1/\tau_k$, meaning the amplifiers must be strong enough to restore amplitudes from photon loss in each channel.
Hence, we require $g_k > 1/\tau_k$ in all modes, i.e., positive solutions with all $x_k>0$~\footnote{Certain input states will admit solutions with non-positive $x_k$ in some modes ($g_k\leq 1/\tau$). In these cases, not all of the amplifier gains can be increased simultaneously, and thus it will not be possible to suppress all relevant noise terms. For example, states of the form $\ket{\psi}=\alpha\ket{m}_1\ket{m+1}_2+\beta\ket{n}_1\ket{n+1}_2$ will satisfy Eq.~\ref{eqn:eigenproblem1} with $g_1 = 1/(\lambda\tau_1)$ and $g_2 = \lambda/\tau_2$. As the eigenvalue $\lambda$ is increased, $g_2$ also increases, but $g_1$ must decrease, even tending toward attenuation $g_1<1$ to keep each mode balanced. Noise terms corresponding to photon loss in mode 1 would therefore not be suppressed relative to the signal.}.
If such solutions exist, then noiseless amplification alone will preserve the input state even for arbitrarily large gain, and thus noiseless attenuation is not required for the NLS protocol.

This system of equations (Eq.~\ref{eqn:eigenproblem2}) can also be expressed with a $L\times M$ matrix multiplying the vector of $x_k$. For the W and NOON states, the values of $n_k$ can be arranged as a diagonal matrix, and hence we can always find a positive solution. For these states, the NLS protocol can be completed without noiseless attenuation. This result also follows for other superpositions of a fixed number of photons $N$ distributed across $M$ modes (see Appendix~\ref{app:equiv}). In contrast, any state with a vacuum term will not have a solution, since one of the rows of this matrix will be the null vector $(0,...,0)$. Similarly, states such as the TMSV with more superposition terms than modes ($L>M$) will likely form an overdetermined set of equations with no solution. For these states, gain alone is not sufficient for NLS, and noiseless attenuation must be included.

In our analysis of W and NOON states, our examples of inputs which did not require noiseless attenuation for NLS, we also observed that the inclusion of noiseless attenuation (when set such that $\nu_{k}\tau_{k}g_{k}=1$) had no effect on the output fidelity, only reducing the probability of success. Here we will see that these properties go hand in hand: states which do not require noiseless attenuation are also those which are (roughly speaking) ``unaffected'' by particular configurations of noiseless attenuation, apart from normalization.

The reason for this is that these input states are not only eigenstates of the \textit{amplifier}-only operator $\hat{X}$, but are also eigenstates of the accompanying \textit{attenuator}-only operator $\hat{Y}$ for some set of $\nu_{k}$.
In particular, it is directly observable from the diagonal form of the lossless transmission operator that when $\nu_{k}\tau_{k}g_{k}=1$ the operator becomes $\hat{X}\hat{Y}=\hat{I}$. Hence $\hat{X}=\hat{Y}^{-1}$, and inverse operators share all the same eigenvectors with inverse eigenvalues.
We note a nuance in this argument: just because all input states are eigenvectors of the identity (when $\nu_{k}\tau_{k}g_{k}=1$ the full $\hat{X}\hat{Y}=\hat{I}$) does not mean that every state is an eigenvector of the individual $\hat{X}$ and $\hat{Y}$. For example, we showed in the main text that the GHZ state (again, as are all states) is an eigenvector of $\hat{X}\hat{Y}=\hat{I}$ despite not being an eigenvector of $\hat{X}$ or $\hat{Y}$ alone, meaning that the full NLS protocol will work on any state but the abbreviated amplifier-only ($\nu_k=1$) protocol will not.
Hence, the following statements are equivalent for a given input state $\ket{\psi}$: (i) NLS can be accomplished without noiseless attenuation (all $\nu_k=1$); (ii) $\ket{\psi}$ is an eigenstate of the \textit{amplifier}-only operator $\hat{X}$ for some set of gains $g_k > 1/\tau_k$; (iii) $\ket{\psi}$ is an eigenstate of the \textit{attenuator}-only operator $\hat{Y}$ for some set of $\nu_k < 1$. 
In App.~\ref{app:equiv} we include a formal proof that when all channels are equivalent these three statements are augmented with a fourth: (iv) $\ket{\psi}$ is a superposition of terms which have the same total photon number. The addition of this fourth statement naturally relates to the W and NOON states considered previously.

\section{Discussion and conclusions}
\label{sec:discussion}

Here, we have extended the analysis of Mi\v{c}uda et al. to various entangled quantum states, including bipartite squeezed vacuum (TMSV), NOON states, and M-partite W and GHZ states. Our results derive the specific NLS conditions required to recover the initial state form supplemented by reducible noise terms, analogous to the results of the original single-mode treatment.
In particular, we found that for both GHZ and TMSV states, NLS can be performed by applying the single-mode NLS protocol to each mode of the entangled state independently.  
However, in a departure from the single-mode results, we showed that noiseless attenuation is not essential for suppressing noise terms for both W and NOON states. This discovery underscores the nuanced role of noiseless attenuation when combined with noiseless amplification and offers a practical example of how GHZ and W state entanglement differ.

The role of noiseless attenuation and NLA in NLS is analogous to that of filters in Procrustean filtering in that noiseless attenuation and NLA each adjust the coefficients of terms in the superposition in order to recover the form of the initial state. Similarly, previous work \cite{ulanov2015undoing}, has shown that amplification alone can recover the form of the initial state with additional vacuum noise terms (though in this case, these noise terms cannot be reduced to zero, as in NLS). What we found is that when a single term in the superposition consists of only the vacuum, NLA alone cannot adjust the weight of this term, as there is nothing to amplify, and hence, in general, for these states (e.g., GHZ and TMSV) the additional degree of freedom to reduce the magnitude of coefficients with noiseless attenuation is required in order to recover the initial state form, but at the cost of a reduced probability of success.

\appendix

\section{Comparing the fidelity and probability of success for different states}\label{app:all_states}

Here we provide plots of the fidelity and probability of success for all four types of states discussed in the paper. These plots show the same results as Figs.~\ref{fig:W_plots}-\ref{fig:NOON_plots} but are arranged for an easier comparison of how the NLS protocol performs for each state. While it is clear higher fidelity comes at the cost of lower probability of success for a given input state, we also observe that this trend is inconsistent when comparing across different input states. For instance, the TMSV has both the highest output fidelity and the highest probability of success, while the NOON state has the lowest of both quantities.

\begin{figure}[t]
    \centering
    \includegraphics[width=\columnwidth]{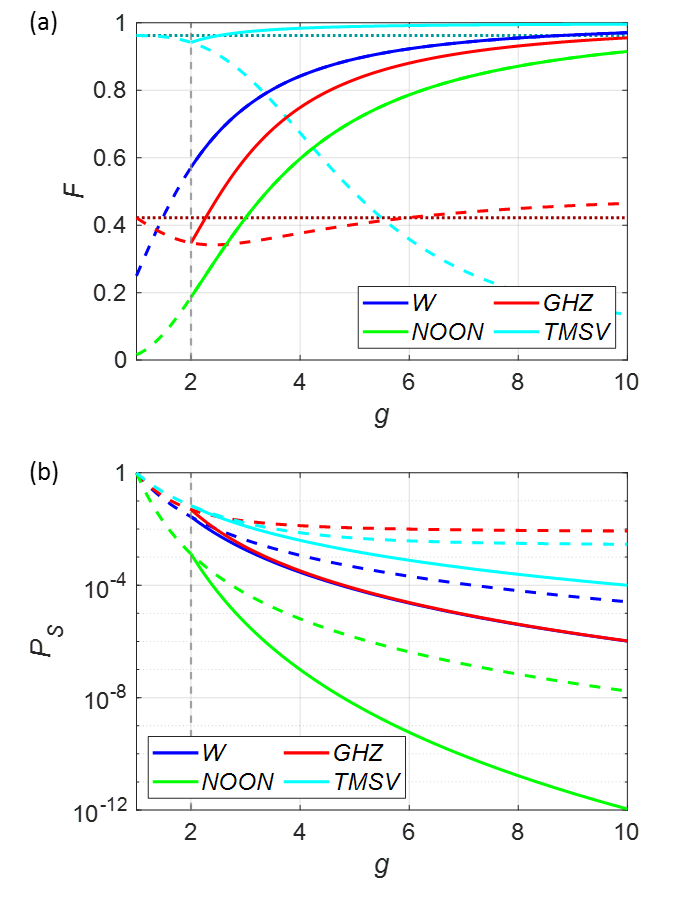}
\caption{Plots of (a) fidelity $F$ and (b) probability of success $P_s$ after NLS with ($\nu=\frac{1}{g\tau}$, solid) and without ($\nu=1$, dashed) noiseless attenuation as a function of gain for all four states discussed in the paper. The darker horizontal dotted lines show the initial fidelity of the GHZ state case (dark red) and TMSV case (dark cyan) without applying the NLS protocol. Channel transmittance $\tau^2=0.25$ for all modes. Note that the regime $g<1/\tau=2$ is unphysical for the case that includes noiseless attenuation.}\label{fig:ALL_plots}
\end{figure}

\section{Simplified forms of expressions}\label{app:general_expr}

We include in this appendix simplifications of the general fidelity and probability of success expressions presented in the main text. In particular, we consider the special case where all channels have the same parameters, including those for noiselesss attenuation and amplification. While the expressions here do not strictly include new information, the simplification can often be dramatic and provide a more clear and intuitive picture of the NLS protocol.

\subsection{W state}
Output fidelity for a $W$ state with $C_W = \nu_l \tau_l g_l$ for all modes $l=1,...,M$:

\begin{align}
    F=&\frac{\bra{W}\hat{\rho}\ket{W}}{\text{Tr}\{\hat{\rho}\}}=\frac{C_W^2}{C_W^2+\frac{1}{M}\sum_{l=1}^{M}\left(1-\tau_{l}^{2}\right)\nu_{l}^2} \\
    =&\frac{C_W^2}{C_W^2+(1-\tau^2)\nu^2} \quad\text{with all channels equal}. \nonumber
\end{align}

The corresponding probability of success:
\begin{align}
    P_s =&  \text{Tr}\{\hat{\rho}\}=\left(\prod_{k=1}^M g_k^{-2}\right)\left(C_W^2+\frac{1}{M}\sum_{l=1}^{M}\left(1-\tau_{l}^{2}\right)\nu_{l}^2\right) \\
    =& g^{-2M}\left[C_W^2+ (1-\tau^2)\nu^2\right] \quad\text{with all channels equal.} \nonumber
\end{align}

\begin{widetext}
\subsection{GHZ state}
Output fidelity for a GHZ state with $g_k = 1/(\nu_k \tau_k)$ for each mode $k=1,...,M$:
\begin{align}
    F=&\frac{\bra{\text{GHZ}}\hat{\rho}\ket{\text{GHZ}}}{\text{Tr}\{\hat{\rho}\}}=\frac{1+\frac{1}{4}\prod_{k=1}^{M}(1-\tau_k^2)\nu_k^2}{1+\frac{1}{2}\left(\sum_{\mathbf{j}\in B^M \setminus {1^M}}\left[\prod_{k=1}^{M}\left((1-\tau_{k}^{2})\nu_k^2\right)^{\neg j_{k}}\right]\right)} \\
    =&\frac{1+\frac{1}{4}[(1-\tau^2)\nu^2]^M}{1+\frac{1}{2}\left(\sum_{\mathbf{j}\in B^M \setminus {1^M}}\left[\prod_{k=1}^{M}\left((1-\tau^{2})\nu^2\right)^{\neg j_{k}}\right]\right)} \quad\text{with all channels equal} \\ \nonumber
    =&\frac{1+\frac{1}{4}[(1-\tau^2)\nu^2]^M}{1+\frac{1}{2}\left[\sum_{j_{tot}=1}^{M}\binom{M}{j_{tot}}\left(\nu^2(1-\tau^2)\right)^{j_{tot}}\right]} \\ \nonumber
    =&\frac{1+\frac{1}{4}[(1-\tau^2)\nu^2]^M}{1+\frac{1}{2}\left[\left(1+\nu^2(1-\tau^2)\right)^{M}-1\right]} \nonumber
\end{align}
where for a binary string $\bm{j}=(j_1,j_2,j_3,...,j_M)$, $j_{tot}=\sum_k{\neg j_k}$ and $j_{tot}$ corresponds to the total number of lost photons. When summing over all the strings $\bm{j}$, there are $\binom{M}{j_{tot}}$ different strings with the same number of lost photons $j_{tot}$, so we can write it as a sum over $j_{tot}$. Note, the sum starts at $j_{tot}=1$ instead of $0$ because string $1^{M}$ was previously factored out.

The corresponding probability of success:
\begin{align}
    P_s =&  \left(\prod_{k=1}^M g_k^{-2}\right)
    \left[1+\frac{1}{2}\left(\sum_{\mathbf{j}\in B^M \setminus {1^M}}\left[\prod_{k=1}^{M}\left((1-\tau_{k}^{2})\nu_k^2\right)^{\neg j_{k}}\right]\right)\right] \\
    =& g^{-2M}\left[1+\frac{1}{2}\left[\left(1+\nu^2(1-\tau^2)\right)^{M}-1\right]\right] \quad\text{with all channels equal.} \nonumber
\end{align}

Without noiseless attenuation $(\nu_k=1)$ and setting $C_{GHZ}=\tau_k g_k$ for $k=1,...M$, the fidelity is given by:
\begin{align}
    F=&\frac{\bra{\text{GHZ}}\hat{\rho}\ket{\text{GHZ}}}{\text{Tr}\{\hat{\rho}\}}=\frac{\frac{(1+C_{\text{GHZ}}^M)^2}{4}+\frac{1}{4}\prod_{k=1}^{M}(1-\tau_k^2)}{\frac{1+C_{\text{GHZ}}^{2M}}{2}+\frac{1}{2}\left(\sum_{\mathbf{j}\in B^M \setminus {1^M}}\left[\prod_{k=1}^{M}C_{\text{GHZ}}^{2j_k}\left(1-\tau_{k}^{2}\right)^{\neg j_{k}}\right]\right)} \\
    =&\frac{\frac{(1+C_{\text{GHZ}}^M)^2}{4}+\frac{1}{4}(1-\tau^2)^M}{\frac{1+C_{\text{GHZ}}^{2M}}{2}+\frac{1}{2}\left[\left(C_{\text{GHZ}}^{2}+1-\tau^2\right)^{M}-C_{\text{GHZ}}^{2M}\right]}  \quad\text{with all channels equal} \nonumber
\end{align}
with a probability of success:
\begin{align}
    P_s =&  \left(\prod_{k=1}^M g_k^{-2}\right)
    \left[\frac{1+C_{\text{GHZ}}^{2M}}{2}+\frac{1}{2}\left(\sum_{\mathbf{j}\in B^M \setminus {1^M}}\left[\prod_{k=1}^{M}C_{\text{GHZ}}^{j_k}\left(1-\tau_{k}^{2}\right)^{\neg j_{k}}\right]\right)\right] \\
    =& \frac{g^{-2M}}{2}\left[1+\left(C_{\text{GHZ}}^2+1-\tau^2\right)^{M}\right] \quad\text{with all channels equal.} \nonumber
\end{align}
\end{widetext}

\subsection{NOON state}
\label{app:NOON}
Output fidelity for a NOON state with $C_N = \nu_l \tau_l g_l$ for all modes $l=1,...,M$:
\begin{align}
    F=&\frac{C_{N}^{2n}}{C_{N}^{2n}+\frac{1}{2}\sum_{l=1}^{2}\sum_{j=1}^{n}\binom{n}{j}C_{N}^{2(n-j)}\left[(1-\tau_{l}^2)\nu_{j}^{2}\right]^j} \\
    =&\frac{C_{N}^{2n}}{[C_{N}^{2}+\nu^2(1-\tau^2)]^{n}} \quad\text{with all channels equal} \nonumber
\end{align}
where the sum in the denominator has been simplified with the binomial theorem. The probability of success is given by:
\begin{align}
    P_s =&  \left(\frac{1}{g_1g_2}\right)^{2n}\left(C_{N}^{2n}+\frac{1}{2}\sum_{l=1}^{2}\sum_{j=1}^{n}\binom{n}{j}C_{N}^{2(n-j)}\left[(1-\tau_{l}^2)\nu_{j}^{2}\right]^j\right) \\
    =& g^{-4n}\left[C_N^2+\nu^2(1-\tau^2)\right]^{n} \quad\text{with all channels equal.} \nonumber
\end{align}

\section{Eigenstate equivalence proof for equal channels}\label{app:equiv}

In section \ref{sec:necessity} we found three equivalent statements, reprinted here for ease of reference: (i) NLS can be accomplished without noiseless attenuation (all $\nu_k=1$); (ii) $\ket{\psi}$ is an eigenstate of the \textit{amplifier}-only operator $\hat{X}$ for some set of gains $g_k > 1/\tau_k$; (iii) $\ket{\psi}$ is an eigenstate of the \textit{attenuator}-only operator $\hat{Y}$ for some set of $\nu_k < 1$. 
Here we will focus on the special case where all channels are equal and show that this allows for a fourth statement to be include: (iv) $\ket{\psi}$ is a superposition of terms which have the same total photon number. The addition of this fourth statement provides an intuitive reason for why the W and NOON state do not require noiseless attenuation in the NLS protocol.

We begin by defining the simplified equal channel operators as
\begin{equation}
    \hat{X}_{g,\tau}=\left(g\tau\right)^{\hat{N}},\hat{Y}_{\nu}=\nu^{\hat{N}}
\end{equation}
where $\hat{N}=\sum_{k=1}^{M}\hat{n}_{k}$ is the total number of photons across all modes.

\begin{theorem}\label{theorem:G}
    A general $M$-mode state $\vert\Phi\rangle=\sum_{l=1}^{L} c_l \ket{n_1 \cdots n_M}^{(l)}$ is an eigenvector of $\hat{X}_{g\tau}$ if and only if each $\ket{n_1 \cdots n_M}^{(l)}$ is an eigenvector of $\hat{N}$ with the same eigenvalue $\hat{N}\ket{n_1 \cdots n_M}^{(l)}=\delta\ket{n_1 \cdots n_M}^{(l)}$ where $\delta$ is a constant.
\end{theorem}
\begin{proof}
First consider the if statement: assume $\hat{N}\ket{n_1 \cdots n_M}^{(l)}=\delta\ket{n_1 \cdots n_M}^{(l)}$ for all $l$ in $\vert\Phi\rangle$. Then 
\begin{equation}
    \hat{X}_{g,\tau}\vert\Phi\rangle=\sum_{l=1}^{L}c_{l}(g\tau)^{\delta}\ket{n_1 \cdots n_M}^{(l)}=(g\tau)^{\delta}\vert\Phi\rangle.
\end{equation}
We can then show this is only if by first assuming ${\hat{X}_{g,\tau}\vert\Phi\rangle=\lambda\vert\Phi\rangle}$ then 
\begin{equation}
    \hat{X}_{g,\tau}\vert\Phi\rangle=\sum_{l=1}^{L}(g\tau)^{N_{l}}c_{l}\ket{n_1 \cdots n_M}^{(l)}=\lambda\sum_{l=1}^{L}c_{l}\ket{n_1 \cdots n_M}^{(l)}
\end{equation}
where $\hat{N}\ket{n_1 \cdots n_M}^{(l)}=N_{l}\ket{n_1 \cdots n_M}^{(l)}$.
Therefore $(g\tau)^{N_{l}}c_{l}=\lambda c_{l}$ for all $l$. For all nonzero $c_{l}$ this then implies $(g\tau)^{N_{l}}=\lambda$ and hence $N_l$ is a constant for all $l$. 
\end{proof}
\begin{corollary}\label{corollary:AL}
A general $M$-mode state $\vert\Phi\rangle=\sum_{l=1}^{L} c_l \ket{n_1 \cdots n_M}^{(l)}$ is an eigenvector of $\hat{Y}_{\nu}$ if and only if each $\ket{n_1 \cdots n_M}^{(l)}$ is an eigenvector of $\hat{N}$ with the same eigenvalue $\hat{N}\ket{n_1 \cdots n_M}^{(l)}=\delta\ket{n_1 \cdots n_M}^{(l)}$ where $\delta$ is a constant.
\end{corollary}
\begin{proof}
The $\hat{Y}_{\nu}= \nu^{\hat{N}}$ operator has the same form as $\hat{X}_{g,\tau}=(g\tau)^{\hat{N}}$ with the replacement of $g\tau$ by $\nu$. The proof therefore follows identically to theorem \ref{theorem:G}.
\end{proof}
\begin{corollary}
    If and only if a state is an eigenvector of $\hat{X}_{g,\tau}$ it is also an eigenvector of $\hat{Y}_{\tau}$.
\end{corollary}
\begin{proof}
If $\vert\Phi\rangle$ is an eigenvector of $\hat{X}_{g,\tau}$ then from theorem \ref{theorem:G} we know $\hat{N}\ket{n_1 \cdots n_M}^{(l)}=\delta\ket{n_1 \cdots n_M}^{(l)}$ and therefore from corollary \ref{corollary:AL} $\vert\Phi\rangle$ is also an eigenvector of $\hat{Y}_{\nu}$. 
The reverse direction follows equivalently, if $\vert\Phi\rangle$ is an eigenvector of $\hat{Y}_{\nu}$ it is also an eigenvector of $\hat{X}_{g,\tau}$.
\end{proof}
Hence, we see that when all channels are equal a state is an eigenvector of both $\hat{X}_{g,\tau}$ and $\hat{Y}_{\nu}$ if and only if it is an eigenvector of one~\footnote{We could have just as well broken $\hat{X}$ into separate operators for gain and loss and had analogous results for three operators without any additional physical insight.}. Further, as part of the proof, we saw the intuitive result that in this special case of equal channels all terms of the eigenvector must have equal total photon number immediately suggesting the W and NOON states -- as we showed in the earlier sections.

\end{document}